\documentclass[11pt]{article}
\usepackage{latexsym,color,amsmath,amssymb,amsthm,fullpage}
\usepackage{graphicx}

\newcommand{\ignore}[1]{}

\newtheorem{theorem}{Theorem}[section]
\newtheorem{lemma}[theorem]{Lemma}

\def\eps{{\varepsilon}}
\def\bd{{\partial}}
\def\A{{\cal A}}

\def\C{{\cal C}}
\def\K{{\cal K}}
\def\R{{\cal R}}
\def\T{{\cal T}}
\def\SS{{\mathbb S}}
\def\reals{{\mathbb R}}

\begin{document}

\begin{titlepage}

\title{Improved Bounds for Geometric Permutations\thanks{%
Work by Haim Kaplan has been supported by Grant 2006/204 from the U.S.-Israel Binational Science Foundation and by Grant 975/06 from the
Israel Science Fund. Work by Micha Sharir was partially supported by
NSF Grant CCF-08-30272, by Grant 2006/194 from the U.S.-Israel Binational
Science Foundation, by grant 338/09 from the Israel Science Fund,
and by the Hermann Minkowski--MINERVA
Center for Geometry at Tel Aviv University. This work is part of the first author's Ph.D. Dissertation prepared under the supervision of the other two authors. A preliminary version of this paper appeared in Proc. 51st Symposium on Foundations of Computer Science, 2010.}}

\author{
Natan Rubin\thanks{%
School of Computer Science, Tel Aviv University, Tel Aviv 69978,
Israel.
E-mail: {\tt rubinnat@post.tau.ac.il }} \and
Haim Kaplan\thanks{%
School of Computer Science, Tel Aviv University, Tel Aviv 69978,
Israel.
E-mail: {\tt haimk@post.tau.ac.il }} \and
Micha Sharir\thanks{%
School of Computer Science, Tel Aviv University, Tel Aviv 69978,
Israel, and Courant Institute of Mathematical Sciences, New York
University, New York, NY 10012, USA. E-mail: {\tt
michas@post.tau.ac.il }} }
\maketitle

\begin{abstract}
We show that the number of geometric permutations of an arbitrary
collection of $n$ pairwise disjoint convex sets in $\reals^d$, for $d\geq 3$,
is $O(n^{2d-3}\log n)$, improving Wenger's 20 years old bound of 
$O(n^{2d-2})$.
\end{abstract}

\maketitle

\end{titlepage}

\section{Introduction} \label{sec:intro}

Let $\K$ be a collection of $n$ convex sets in $\reals^d$. 
A line $\ell$ is a \emph{transversal} of $\K$ if it intersects 
all the sets in $\K$.
If the objects in $\K$ are \emph{pairwise disjoint}, an oriented
line transversal meets them in a well-defined order, called a 
{\em geometric permutation}. The study of geometric permutations 
plays a central role in geometric transversal theory; 
see \cite{GPW1,Wen-surv} for compherensive surveys.

\medskip
\noindent{\bf Previous work.}
In 1985, Katchalski et al.~\cite{KLZ} initiated the study 
of the maximum possible number $g_d(n)$ of geometric permutations 
induced by a set $\K$ of $n$ pairwise disjoint convex 
objects in $\reals^d$.  They constructed, for any $n\ge 4$,
a family of $n$ pairwise disjoint convex sets in $\reals^2$ 
that admits $2n-2$ geometric permutations.  
Edelsbrunner and Sharir \cite{ES} showed, five 
years later, that this bound is tight in the worst case,
implying that $g_2(n)=2n-2$. Wenger \cite{Wen2} proved, also 
in 1990, that $g_d(n)=O(n^{2d-2})$ in any dimension $d\geq 3$.
In 1992, Katchalski et al.~\cite{KLL} generalized their lower 
bound construction and showed that there exist collections of 
$n$ pairwise disjoint convex sets in $\reals^d$, for any $d\geq 3$, which admit $\Omega(n^{d-1})$ 
geometric permutations.  Since then, closing (or even reducing) 
the fairly large gap between these upper and lower bounds 
on $g_d(n)$, in any dimension $d\geq 3$, has remained one
of the major long standing open problems in geometric 
transversal theory.

Several partial steps towards this goal were made in the past 
decade. Most of them deal with geometric permutations of 
certain restricted families of pairwise disjoint convex bodies 
in $\reals^d$.  For example, Smorodinsky et al.~\cite{SMS} 
derived a tight upper bound of $\Theta(n^{d-1})$ on the number of 
geometric permutations induced by an arbitrary collection of $n$ 
pairwise disjoint balls in $\reals^d$.  
Katz and Varadarajan \cite{KV} generalized this result to 
arbitrary collections of $n$ pairwise disjoint {\em fat} 
convex bodies.  Other recent works \cite{CGN,HXC,KSZ,ZS} show 
that the maximum possible number of geometric permutations 
induced by pairwise disjoint {\em unit} balls (or, more generally, 
balls of bounded size disparity) is constant in any 
dimension. 

Other studies bound the number of geometric permutations induced by 
arbitrary collections of pairwise disjoint convex sets, whose 
realizing transversal lines belong to some restricted subfamily
of lines in $\reals^d$. For example, Aronov and Smorodinsky \cite{ArSm} 
derive a tight bound of $\Theta(n^{d-1})$ on the maximum number of 
geometric permutations realized by lines that pass through a 
fixed point in $\reals^d$.  A recent paper \cite{KRS} by the 
authors studies line transversals of arbitrary convex polyhedra 
in $\reals^3$ and derives (as a byproduct) an improved upper 
bound of $O(n^{3+\eps})$, for any $\eps>0$, on the number of 
geometric permutations realized by lines which pass through a 
fixed line in $\reals^3$. 

\medskip
\noindent{\bf The space of line transversals.}
Lines in $\reals^d$ have $2d-2$ degrees of freedom, and are 
naturally represented in a real projective space (so-called the 
{\em Grassmannian manifold}; see~\cite{GPW1}). 
However, for the purpose of combinatorial analysis, we can 
represent them (with the exclusion of some ``negligible'' subset 
which we may ignore) by points in the real Euclidean space
$\reals^{2d-2}$; see \cite{GPW1} for more details.  

Let $\K$ be a collection of $n$ convex sets in $\reals^d$, not
necessarily pairwise disjoint.
The \emph{transversal space} $\T(\K)$ of $\K$ is the set in
$\reals^{2d-2}$ of all (points representing) the transversal 
lines of $\K$.

If the sets of $\K$ are pairwise disjoint then any two lines in the 
same connected component of $\T(\K)$ induce the same geometric 
permutation, so the number of geometric permutations is upper bounded
by the number of components of $\T(\K)$.  In two dimensions, the 
converse property also holds. That is, lines that stab $\K$ in a 
fixed order form a single connected component of $\T(\K)$; 
see, e.g., \cite{GPW}. Thus, according to \cite{ES}, the 
transversal space $\T(\K)$, of any family $\K$ of $n$ pairwise 
disjoint convex sets in $\reals^2$ has at most $2n-2$ 
connected components.

The situation becomes considerably more complicated already in 
$\reals^3$: There exist collections of four (pairwise disjoint) 
convex sets whose transversal space consists of an arbitrarily 
large number of connected components \cite{GPW,KRS}.
This is a simple instance of the phenomenon that the shape
of $\T(\K)$ depends on the shape of the sets in $\K$, and may 
grow out of control if we do not impose any restrictions on the 
sets of $\K$.
This might explain (in part) the difficulty of extending 
the relatively simple analysis of the number of geometric 
permutations in $\reals^2$ to higher dimensions. 

In three dimensions, if the sets in $\K$ have {\em constant description complexity} 
(i.e., each set can be described as a Boolean combination of a
constant number of polynomial equalities and inequalities of constant
maximum degree) then one can obtain sharp bounds on the combinatorial
complexity of $\T(\K)$ 
(see \cite{KS,Wen-surv} for a precise definition).
Specifically, the analysis of Koltun and 
Sharir \cite{KS} yields an improved bound of $O(n^{3+\eps})$, 
for any $\eps>0$, on the combinatorial complexity, and thus also on
the number of connected components of $\T(\K)$, for collections $\K$
of this kind. (If $\K$ is a collection of $n$ triangles in $\reals^3$, an improved bound of $O(n^3\log n)$ holds, see \cite{Pankaj}.) Hence, this also serves as an upper bound on the 
number of geometric permutations induced by any such collection $\K$.
Using this approach, and continuing to assume that
the sets in $\K$ have constant description complexity, 
one can strengthen Wenger's bound of $O(n^{2d-2})$~\cite{Wen2} to
apply to the combinatorial complexity of $T(\K)$, and not just to
the number of geometric permutations. The strength (and beauty) of
Wenger's analysis is that it yields this bound without making any 
assumptions whatsoever on the shape of the sets in $\K$ 
(other than being convex and pairwise disjoint).

\medskip
\noindent{\bf Our results.}
We first show that the number of geometric permutations admitted by 
\emph{any} collection of $n$ pairwise disjoint convex sets 
in $\reals^3$ is $O(n^3\log n)$, thus improving Wenger's previous 
upper bound on $g_3(n)$ roughly by a factor of $n$.
Our approach can be generalized to higher dimensions, and yields 
an improved upper bound of $O(n^{2d-3}\log n)$ on $g_d(n)$, 
for any $d\geq 3$. (In three dimensions, our bound is also a slight improvement of 
the bound $O(n^{3+\eps})$, for any $\eps>0$, of \cite{KS} 
for the case where the sets in $\K$ have constant description 
complexity.) 

Here is a brief overview of our solution in $\reals^3$.
Following the approach of Wenger \cite{Wen2}, we represent the 
directions of transversal lines by points on the unit $2$-sphere 
$\SS^2$, separate every pair of objects in $\K$ by a plane, and
associate with each such plane the great circle on $\SS^2$ parallel 
to it. We then consider the arrangement $\A$ of the resulting
${n\choose 2}$ great circles on $\SS^2$, which
consists of $O(n^4)$ 2-faces. The crucial observation made in 
\cite{Wen2} is that all transversal lines, whose directions belong to the same 2-face of $\A$, stab the sets of $\K$ in the same order (if the face contains such directions at all).  Hence, the number of geometric 
permutations is upper bounded by the total number of 2-faces of 
$\A$, implying that $g_3(n)=O(n^4)$.

We improve this bound by showing that the actual number of faces 
which contain at least one direction of a transversal line 
(so-called {\em permutation faces}) is only $O(n^3\log n)$.  
Moreover, we show that the overall number of edges and vertices 
on the boundaries of these faces is also at most $O(n^3\log n)$.

The analysis proceeds in two steps.
First, we use a direct geometric analysis to show that the number 
of vertices whose four incident faces are all permutation faces 
is $O(n^3)$. We refer to such vertices as \emph{popular vertices}.
Informally, we associate with each popular vertex $v$ (with 
the possible exception of $O(n^3)$ ``degenerate'' ones) the 
intersection line $\lambda_v$ of the two separating planes $h,h'$
that correspond to the two circles incident to $v$, and show that
$\lambda_v$ stabs exactly $n-4$ sets of $\K$ (all but the sets in the 
two pairs separated by $h$ and $h'$, respectively).  We then apply,
within each of the ${n \choose 2}$ separating planes, the linear 
bound on the number of geometric permutations in $\reals^2$, 
due to Edelsbrunner and Sharir \cite{ES}, combined with a simple
application of the Clarkson-Shor probabilistic analysis
technique~\cite{CS}, and thereby obtain the overall $O(n^3)$ 
asserted bound on the number of popular vertices. 

We then use this bound to analyze the overall number of vertices incident to permutation faces.
This is achieved by a refined (and simplified) variant of the 
charging scheme of Tagansky \cite{Tagansky}. 

The analysis can be extended to any dimension $d\geq 4$, but its technical details become somewhat more involved.

The paper is organized as follows. We first derive the
nearly-cubic upper bound on $g_3(n)$. To this end, we begin in
Section~\ref{Sec:Prelim} by introducing some notations and the infrastructure, and then
establish this bound in Section~\ref{Sec:Perms3D}. In
Section \ref{Sec:Higher}, we extend the analysis to any dimension
$d\ge 4$. 

\section{Preliminaries}\label{Sec:Prelim}

\noindent{\bf The setup in $\reals^3$.}
Let $\K$ be a collection of $n$ arbitrary pairwise disjoint
convex sets in $\reals^3$. We may also assume, without loss of 
generality, that the elements of $\K$ are \emph{compact}. Indeed,  
let $g_d(n)$ be the maximum possible number of geometric
permutations induced by a collection of $n$ pairwise disjoint
\emph{compact} convex sets in $\reals^3$. Let $\K=\{K_1,\ldots,K_n\}$ be a collection
of $n$ \emph{arbitrary} pairwise disjoint convex sets in $\reals^3$,
which induces $m$ geometric permutations, realized by $m$ respective
lines $\ell_1,\ldots,\ell_m$. For each $1\leq i\leq n$ and 
$1\leq j\leq m$, let $p_{ij}$ denote an arbitrary point in 
$K_i\cap \ell_j$. For each $1\leq i\leq n$, let $K'_i$ denote the 
convex hull of $\{p_{ij}\mid 1\leq j\leq m\}$, and observe that 
$K'_i$ is a compact convex subset of $K_i$. Hence
$\K'=\{K'_1,\ldots,K'_n\}$ is a collection of $n$ pairwise disjoint
\emph{compact} convex sets, which induces (at least) $m$ geometric
permutations (realized by the same lines $\ell_1,\ldots,\ell_m$), 
so $m\leq g_d(n)$.

We use the
following setup, introduced by Wenger \cite{Wen2} and briefly mentioned in the introduction, 
to analyze geometric permutations of $\K$. Enumerate the elements of $\K$ as $K_1,K_2,\ldots,K_n$. 
For each $1\leq i<j\leq n$ we fix some plane $h_{ij}$ which strictly
separates $K_i$ and $K_j$. 
We orient $h_{ij}$ so that $K_i$ lies 
in the open negative halfspace $h_{ij}^-$ that it bounds, and $K_j$ 
lies in the open positive halfspace $h_{ij}^+$.
We represent directions of (oriented) lines in $\reals^3$ by points on the 
unit 2-sphere $\SS^{2}$. Without loss of generality we may assume 
that the planes $h_{ij}$ are in {\em general position}, meaning 
that every triple of them intersect at a single point,
and no four meet at a common point.

Each separating plane $h_{ij}$ induces a great circle $C_{ij}$ on 
$\SS^2$, formed by the intersection of $\SS^2$ with the plane 
parallel to $h_{ij}$ through the origin. Equivalently, $C_{ij}$ 
is the locus of the directions of all lines parallel to $h_{ij}$.
$C_{ij}$ partitions $\SS^{2}$ into two open hemispheres 
$C_{ij}^+$, $C_{ij}^-$, so that $C_{ij}^+$ (resp., $C_{ij}^-$)
consists of the directions of lines which cross $h_{ij}$ from 
$h_{ij}^-$ to $h_{ij}^+$ (resp., from $h_{ij}^+$ to $h_{ij}^-$). 
Note that lines whose directions lie in $C_{ij}$ cannot stab 
both $K_i$ and $K_j$.  Thus, any oriented common transversal 
line of $K_i$ and $K_j$ intersects $K_j$ after (resp., before) $K_i$ 
if and only if its direction lies in $C_{ij}^+$ (resp., $C_{ij}^-$).

Put $\C(\K)=\{C_{ij}\mid 1\leq i<j\leq n\}$, and consider the 
arrangement $\A(\K)$ of the $n\choose 2$ great circles of $\C(\K)$. 
The assumption that the planes $h_{ij}$ are in general position 
is easily seen to imply that the circles in $\C(\K)$ are also 
in general position, in the sense that no pair of them coincide
and no three have a common point.  Each 2-face $f$ of $\A(\K)$ 
induces a relation $\prec_f$ on $\K$, in which $K_i\prec_f K_j$ 
(resp., $K_j\prec_f K_i$) if $f\subseteq C_{ij}^+$ (resp., 
$f\subseteq C_{ij}^-$).  Clearly, the direction of each oriented 
line transversal $\lambda$ of $\K$ belongs to the unique 2-face $f$
of $\A(\K)$ whose relation $\prec_f$ coincides with the order 
in which $\lambda$ visits the sets of $\K$ (as noted above, the 
direction of $\lambda$ cannot lie on an edge or at a vertex of 
$\A(\K)$). In particular, the number of geometric permutations 
is bounded by the number of 2-faces of $\A(\K)$, which is $O(n^{4})$.

This is the way in which Wenger established this upper bound (in 
three dimensions) 20 years ago~\cite{Wen2}. Moreover, this
approach can be extended to any dimension $d\ge 3$, and yields 
the upper bound $O(n^{2d-2})$ on $g_d(n)$; see \cite{Wen2} and
Section~\ref{Sec:Higher} below. 
The main weakness of this
argument (as follows from the analysis in this paper) is that most
faces of $\A(\K)$ do not induce a geometric permutation of $\K$.
Specifically, for some faces $f$ the relation $\prec_f$ might have cycles, in which case $f$ clearly cannot contain the direction of a transversal of $\K$. But even if $\prec_f$ is acyclic (and thus a total order) there need not exist any line transversal with direction in $f$.


\medskip
\noindent{\bf More definitions.} 
We need a few more notations.  We call a 2-face of $\A(\K)$ a 
\emph{permutation face} if there is at least one line transversal 
of $\K$ whose direction belongs to $f$.  Note, however, that
the directions of the line transversals of $\K$ within a fixed
permutation face $f$ is only a subset of $f$, which need not even 
be connected; see, e.g., a construction in \cite{GPW} and the
introduction.

Each pair of great circles of $\C(\K)$ intersect at exactly two 
antipodal points of $\SS^2$. By the general position assumption, 
all the circles are distinct, and each vertex $v$ of $\A(\K)$ is 
incident to exactly two great circles.  Hence, each vertex is 
incident to exactly four (distinct) faces of $\A(\K)$. Assuming that
$|\K|\geq 3$, $\C(\K)$ contains at least three great circles, so the
boundary of each cell of $\A(\K)$ contains at least three vertices.
This, and the fact that each vertex is incident to four faces, imply
that the number of permutation faces in $\A(\K)$ is at most 
proportional to the overall number of their vertices. It is this latter
quantity that we proceed to bound.

We say that vertex $v$ in $\A(\K)$ is \emph{regular} if the two
great circles $C_{ij}$, $C_{k\ell}$ incident to $v$ are defined 
by four \emph{distinct} sets of $\K$; otherwise, when only three 
of the indices $i,j,k,\ell$ are distinct, we call $v$ a 
{\em degenerate} vertex.  Clearly, the number of degenerate 
vertices is $O(n^{3})$, so it suffices to bound the number of 
regular vertices of permutation faces.

In the forthcoming analysis we will use subcollections $\K'$ of $\K$,
typically obtained by removing one set, say $K_q$, from $\K$. Doing so
eliminates all separating planes $h_{iq}$, for $i=1,\ldots,q-1$, and
$h_{qi}$, for $i=q+1,\ldots,n$. Accordingly, the corresponding circles
$C_{iq}$, $C_{qi}$ are also eliminated from $\C(\K')$, and $\A(\K')$
is constructed only from the remaining circles. In particular, a
regular vertex $v$ of $\A(\K)$, formed by the intersection of 
$C_{ij}$ and $C_{k\ell}$, remains a 
vertex of $\A(\K')$ if and only if $q\ne i,j,k,\ell$. An edge 
(resp., face) of $\A(\K')$ may contain several edges (resp., faces) 
of $\A(\K)$. Note that if $f'$ is a face of $\A(\K')$ which contains 
a permutation face $f$ of $\A(\K)$ then $f'$ is a permutation face 
in $\A(\K')$; the permutation that it induces is the permutation 
of $f$ with $K_q$ removed.

\section{The Number of Geometric Permutations in $\reals^3$}
\label{Sec:Perms3D}

Our analysis uses the setup of Tagansky \cite{Tagansky},
somewhat adapted to our context.  To make this paper more 
self-contained, we will spell out many of the
details of the technique as it applies in our context.

\medskip
\noindent{\bf Popular vertices and edges.}
We say that an edge $e$ of $\A(\K)$ is {\em popular} if its
two incident faces are both permutation faces. We say that a vertex
$v$ of $\A(\K)$ is {\em popular} if its four incident faces are all
permutation faces.  We establish the upper bound $O(n^3)$ on the 
number of popular vertices, using a direct geometric 
argument.  The analysis then proceeds by applying two charging
schemes. 
The first scheme results in a recurrence which expresses 
the number of popular edges in terms of the number of popular 
vertices. The second scheme leads to a recurrence which expresses 
the number of vertices of permutation faces in terms of the number 
of popular edges. The solutions of both recurrences are nearly cubic.
Naive (and simpler) implementation of both schemes incurs an 
extra logarithmic factor in each recurrence, resulting in the
overall bound $g_3(n) = O(n^3\log^2n)$. With a more careful analysis
of the second scheme, we are able to eliminate one of these factors,
and thus obtain the bound $g_3(n) = O(n^3\log n)$. 

\subsection{The number of popular vertices}
\label{subsec:pop3D}

For a regular vertex $v$ of $\A(\K)$, formed by the intersection 
of $C_{ij},C_{k\ell}\in \C(\K)$, we denote by $\K_v$ the collection 
$\{K_{i},K_{j},K_{k},K_{\ell}\}$ of the four sets defining (the circles meeting at) $v$.

\begin{lemma}\label{Thm:ConsecutivePairs}
Let $v$ be a regular popular vertex of
$\A(\K)$, incident to $C_{ij},C_{k\ell}\in \C(\K)$. \\
(i) Each pair of sets $K_a\in \K_v$ and $K_b\in \K\setminus\K_v$ 
appear in the same order in all four permutations induced by the 
faces incident to $v$.\\
(ii) The elements of each pair $K_i,K_j$ and $K_k,K_\ell$ are 
consecutive in all four permutations induced by the faces 
incident to $v$.
\end{lemma}  
\begin{proof}
Any two distinct faces $f,g$ incident
to $v$ are separated only by one or two great circles from
$\{C_{ij},C_{k\ell}\}$, so the orders $\prec_f$ and $\prec_g$
may disagree only over the pairs $(K_i,K_j)$ and $(K_k,K_\ell)$.
As a matter of fact, the four permutations are obtained from each
other only by swapping $K_i$ and $K_j$ and/or swapping $K_k$ and 
$K_\ell$. This is easily seen to imply both parts of the lemma.
\end{proof}

\begin{lemma}\label{Thm:PopularLine}
Let $v$ be a regular popular vertex in
$\A(\K)$, incident to $C_{ij},C_{k\ell}\in \C(\K)$. Then the line $\lambda_v=h_{ij}\cap h_{k\ell}$ stabs all
the $n-4$ sets in $\K\setminus\K_v$, and misses all four sets in $\K_v$.
\end{lemma}
\begin{proof}
By definition, $\lambda_v$ misses every set $K\in \K_v$, because 
it is contained in a plane separating $K$ from another set in 
$\K_v$.  Hence, it suffices to show that $\lambda_v$ is a 
transversal of $\K\setminus\K_v$.

To show this, we fix a set $K_a\in\K\setminus\K_v$ and show that
each of the four dihedral wedges determined by $h_{ij}$ and $h_{k\ell}$
meets $K_a$. The convexity of $K_a$ then implies that $\lambda_v$ 
intersects $K_a$; see Figure \ref{fourq} (left).

\begin{figure}[htb]
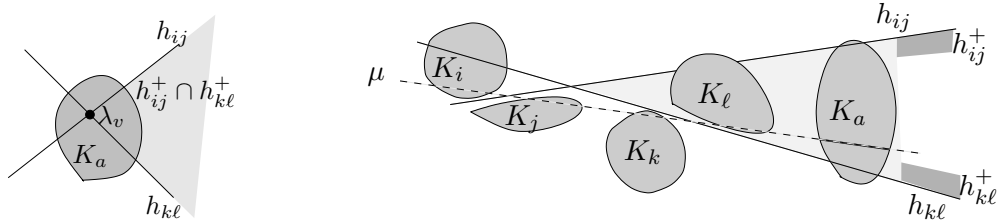

\begin{center}
\input{FourWedges.pstex_t}\hspace{2cm}\input{CrossSide.pstex_t}
\caption{\small\sf
Left: $K_a$ must cross $\lambda_v=h_{ij}\cap h_{k\ell}$ 
since it meets each of the four incident wedges (one of which is
highlighted). 
Right: The transversal line $\mu$ crosses $K_a$ after 
$K_i,K_j,K_k,K_\ell$, so the segment $K_a\cap \mu$ (highlighted) is contained 
in $h_{ij}^+\cap h_{k\ell}^+$.} 
\label{fourq}
\end{center}
\vspace*{-0.8cm}
\end{figure}

Lemma~\ref{Thm:ConsecutivePairs} implies that $K_a$ lies at the 
same position in each of the four permutations induced by the 
faces incident to $v$. Without loss of generality, assume that 
the consecutive pair $K_i,K_j$ appears in these permutations 
before the consecutive pair $K_k,K_\ell$.
Then either $K_a$ precedes both pairs in all four permutations, 
or appears in between them, or succeeds both of them. In what 
follows we assume that $K_a$ succeeds both pairs in all the 
permutations, but similar arguments handle the other two cases too.

Consider the permutation $\pi_1$ induced by the face $f_1$ incident 
to $v$ and lying in $C_{ij}^+\cap C_{k\ell}^+$, and let $\mu$ be a 
line transversal which induces $\pi_1$. Since the direction of
$\mu$ lies in $C_{ij}^+\cap C_{k\ell}^+$, it follows that $\mu$
crosses $h_{ij}$ from the side containing $K_i$ to the side 
containing $K_j$, and it crosses $h_{k\ell}$ from the side 
containing $K_k$ to the side containing $K_\ell$. 
Hence $K_i$ precedes $K_j$ and $K_k$ precedes $K_\ell$ 
in $\pi_1$. Moreover, $\mu$ crosses $h_{ij}$ in between its 
intersections with $K_i$ and $K_j$, and it crosses $h_{k\ell}$ 
in between its intersections with $K_k$ and $K_\ell$. Thus, 
$\mu\cap K_a$ lies in $h_{ij}^+\cap h_{k\ell}^+$; see 
Figure~\ref{fourq} (right). That is, $K_a$ intersects the dihedral wedge 
$h_{ij}^+\cap h_{k\ell}^+$. Fully symmetric 
arguments, applied to the permutations induced by the three 
other faces $f_2,f_3,f_4$ incident to $v$, show that $K_a$ intersects each of 
the three other dihedral wedges determined by $h_{ij}$ and 
$h_{k\ell}$, which, as argued above, implies 
that $\lambda_v$ stabs $K_a$. 
As promised, slightly modified variants of this argument 
(with different correspondences between the wedges around 
$\lambda_v$ and the faces around $v$) handle the cases where 
$K_a$ precedes both pairs $K_i,K_j$ and $K_k,K_\ell$ in all 
four permutations, or appears in between these pairs.
\end{proof}

\begin{theorem}\label{Thm:NZero}
Let $\K$ be a collection of $n$ pairwise disjoint compact convex sets 
in $\reals^3$. Then the number of popular vertices in $\A(\K)$ is 
$O(n^{3})$. 
\end{theorem}
\begin{proof}
Note first that each popular vertex must be regular. Indeed,
if $v$ is a degenerate popular vertex incident to, say, $C_{ij},C_{ik}\in \C(\K)$, then, arguing as in Lemma \ref{Thm:ConsecutivePairs}, each of the two pairs $K_i,K_j$ and $K_i,K_k$ appears consecutively in each of the four permutations near $v$. 
Let $f$ be one of the four permutation faces incident to $v$, and assume, without loss of generality, that
$K_k\prec_f K_i\prec_f K_j$. Let $g$ be the permutation face neighboring to $f$ and separated from it only by the circle $C_{ik}$.
Then we must have $K_i\prec_g K_k \prec_g K_j$, contradicting the fact that $K_i,K_j$ are consecutive also under $\prec_g$.

Now let $v$ be a regular popular vertex in $\A(\K)$, incident to $C_{ij},C_{k\ell}\in \C(\K)$,   
and let $\lambda_v=h_{ij}\cap h_{k\ell}$ be the line considered
in Lemma~\ref{Thm:PopularLine}. 
Put $K^*_q = K_q\cap h_{ij}$, 
for each index $q\ne i,j$, and denote by $\K^*$ the collection 
of these $n-2$ planar cross-sections within $h_{ij}$.  Clearly, 
all sets in $\K^*$ are pairwise disjoint, compact, and convex.

\begin{figure}[htb]
\begin{center}
\input{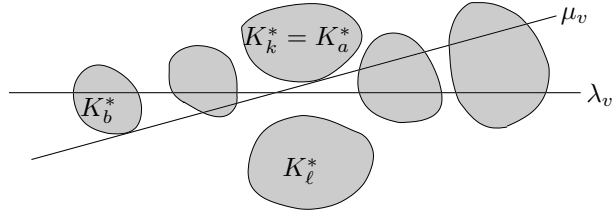}
\caption{\small\sf
View inside $h_{ij}$: The line $\lambda_v=h_{ij}\cap h_{k\ell}$ 
misses $K_k^*,K_\ell^*$ but stabs all other sets in $\K^*$. The
line $\mu_v$ is tangent to $K_a^*=K^*_k$ and to $K_b^*$, so it misses only $K_\ell^*$.} 
\label{Fig:AlmostTransversal}
\end{center}
\vspace*{-0.6cm}
\end{figure}

By Lemma~\ref{Thm:PopularLine}, $\lambda_v$ lies in $h_{ij}$, stabs all 
the sets in $\K^*\setminus\{K^*_k,K^*_\ell\}$ (so they are all
nonempty) and misses the two sets $K_k^*,K_\ell^*$. (As can be 
easily verified, both of $K_k^*,K_\ell^*$ are also nonempty, 
although our analysis does not rely on this property.)

Translate $\lambda_v$ within $h_{ij}$ until it becomes tangent to 
some set $K_a^*\in \K^*$, and then rotate the resulting line 
around $K_a^*$, say counterclockwise,
keeping it tangent to that set, until it becomes 
tangent to another set $K_b^*\in \K^*\setminus\{K_a^*\}$. 
The sets $K_k^*,K_\ell^*,K_a^*,K_b^*$ need not all be distinct, so 
the resulting extremal tangent $\mu_v$ misses {\em at most} two sets
of $\K^*$ and intersects all the other sets; see
Figure~\ref{Fig:AlmostTransversal}.

We charge $\lambda_v$ to $\mu_v$, and argue that each extremal line 
$\mu$ in $h_{ij}$, which is tangent to two sets of $\K^*$ and misses 
at most two other sets of $\K^*$, is charged in this manner at most 
twice.  Indeed, by the general position assumption, $\mu$ lies in 
a single plane $h_{ij}$. Within that plane, if $\mu$ misses two 
sets of $\K^*$ then these must be the sets $K_k^*,K_\ell^*$. If $\mu$ 
misses only one set of $\K^*$ then this set must be one of the sets 
$K_k^*,K_\ell^*$, and the other set is one of the two sets $\mu$ is
tangent to. Finally, if $\mu$ does not miss any set of $\K^*$ then 
$K_k^*,K_\ell^*$ are the two sets $\mu$ is tangent to. Hence $\mu$ 
determines at most two quadruples $K_i,K_j,K_k,K_\ell$ whose lines 
$\lambda_v$ can charge $\mu$, and the claim follows.

It therefore suffices to bound the number of extremal lines $\mu$ 
charged in this manner.
This can be done using the Clarkson-Shor technique \cite{CS}, by 
observing that each such line $\mu$ is defined by two sets of $\K^*$ 
(those it is tangent to; any such pair of sets determine four
common tangents) and is ``in conflict'' with at most two other 
sets (those that it misses). Thus, the Clarkson-Shor technique
implies that the number of lines $\mu_v$ is $O\left(L_0(n/2)\right)$,
where $L_0(r)$ is the (expected) number of extremal lines which are
transversals to a (random) sample of $r$ sets of $\K^*$.
Edelsbrunner and Sharir \cite{ES} establish an upper bound of 
$O(r)$ on the complexity of the space of line transversals to a 
collection of $r$ pairwise-disjoint compact convex sets in the plane, 
implying that $L_0(r)=O(r)$.  Hence the number of charged lines $\mu$ 
in a single plane $h_{ij}$ is $O(n)$, for a total of 
$O\left({n\choose 2}\cdot n\right)=O(n^3)$. Since, as noted above, 
each line is charged at most twice in its plane, this also bounds the number of
popular vertices.
\end{proof}

\subsection{The number of popular edges}
\label{Subsec:PopularEdges}

We next bound the number of popular edges in $\A(\K)$, using the bound
on popular vertices just derived.  We define an \emph{edge border} 
in $\A(\K)$ to be a pair $(v,Q)$, where $v$ is a vertex of $\A(\K)$, 
incident to two great circles $C_{ij},C_{k\ell}$, and $Q$ is one of 
the four open hemispheres $C_{ij}^-,C_{ij}^+,C_{k\ell}^-,C_{k\ell}^+$ 
determined by one of these circles. See Figure \ref{Fig:EdgeBorder} (left).  Note that $Q$ determines a 
unique edge $e$ of $\A(\K)$ which is incident to $v$ and is contained 
in $Q$.  If, in addition, $e$ is a popular edge, we say that $(v,Q)$ 
is a {\em popular edge border}. For the purpose of the analysis, we
will also refer to $(v,Q)$ as a {\em $0$-level edge border}.

\begin{figure}[htb]
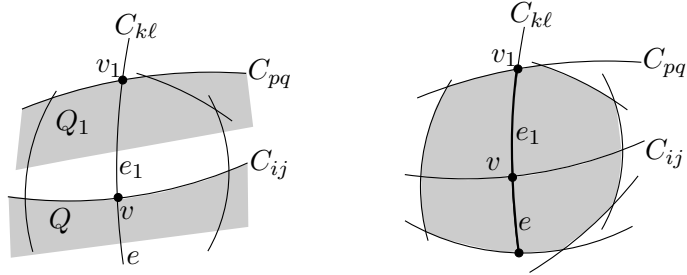

\begin{center}
\input{EdgeBorder.pstex_t}\hspace{2cm}\input{TwoEdges.pstex_t}
\caption{\small\sf
Left: Charging a $0$-level edge border $(v,Q)$ to a $1$-level edge 
border $(v_1,Q_1)$. Right: If the edges $e,e_1$ are both popular then 
$v$ is a popular vertex.} 
\label{Fig:EdgeBorder}
\end{center}
\vspace*{-0.8cm}
\end{figure}

One useful 
feature of the border notation is that if $(v,Q)$ is an edge border 
in $\A(\K)$ and $\K'$ is a subcollection of $\K$ so that $v$ 
is still a vertex of $\A(\K')$, then $(v,Q)$ is also an edge border
in $\A(\K')$. The edge $e'$ of $\A(\K')$ associated
with $(v,Q)$ in $\A(\K')$ either is equal to $e$, or strictly contains 
$e$ (in the latter case both $e$ and $e'$ have $v$ as a common
endpoint).

If an edge border $(v,Q)$, which is not a $0$-level edge border, becomes 
a $0$-level edge border after removing from $\K$ some single set $K_a\in \K$, we call it a 
{\em $1$-level edge border}. In this case we say that $(v,Q)$ is 
\emph{in conflict} with $K_a$.  Note that the set $K_a$, whose 
removal makes $(v,Q)$ a $0$-level edge border, need not be unique; 
see Section \ref{Subsec:PermutFaces} for further discussion. 

Clearly, to bound the number of popular edges it suffices to bound 
the number of $0$-level edge borders, which is twice the number of 
popular edges in $\A(\K)$ (each edge is counted once at each of its
endpoints).

Since each vertex of $\A(\K)$ participates in exactly 
four edge borders, the number of edge borders which are incident to a 
degenerate vertex is $O(n^3)$.  We bound the number of remaining 
$0$-level edge borders using the following charging scheme. 

Let $(v,Q)$ be a $0$-level edge border, where $v$ is incident to 
$C_{ij}$ and $C_{k\ell}$, so that $Q=C_{ij}^+$, say. Let 
$e$ be the popular edge associated with $(v,Q)$. Trace $C_{k\ell}$ 
from $v$ away from $e$ (into $C_{ij}^-$), and let $v_1$ be the next 
encountered vertex. Let $e_1$ be the edge connecting $v$ and $v_1$. 
Let $C_{pq}$ be the other circle incident to $v_1$ and assume, without
loss of generality, that 
$v$ lies in $C_{pq}^+$. See Figure \ref{Fig:EdgeBorder} (left). 
Note that, assuming $|\K|\ge 3$, we have
$C_{pq}\neq C_{ij}$ (i.e., $v_1$ is not antipodal to $v$), because 
otherwise $C_{ij}$ would have intersected only $C_{k\ell}$.
One of the following cases must arise:

\smallskip
(i) 
$v_1$ is degenerate.

\smallskip
(ii)
The edge $e_1$ is also popular, so $v$ is a popular vertex;
see Figure \ref{Fig:EdgeBorder} (right).

\smallskip
(iii) 
$e_1$ is not popular. Since $C_{pq}\neq C_{ij}$, one of $i,j$, say $i$, 
is different from both $p$ and $q$.  This (and the fact that 
$i\ne k,\ell$) implies that removing $K_i$ from $\K$ also removes 
$C_{ij}$ from $\A$, keeps $v_1$ intact, and makes the appropriate 
extension of $e$ reach (and terminate at) $v_1$, thereby making 
$(v_1,Q_1)$ a $0$-level edge border in $\A(\K\setminus\{K_i\})$, 
where $Q_1 = C_{pq}^+$. See Figure \ref{Fig:EdgeBorder} (left).

\smallskip
In case (i) we charge $(v,Q)$ to $v_1$. The number of degenerate
vertices is $O(n^3)$ and each of them can be charged only $O(1)$ times
in this manner. Hence, the number of
$0$-level edge borders that fall into this subcase is $O(n^3)$.

In case (ii) we can charge $(v,Q)$ to $v$. Since a popular vertex 
participates in exactly four $0$-level edge borders, the number of 
$0$-level edge borders that fall into this subcase is $O(n^3)$, by
Theorem~\ref{Thm:NZero}. 

In case (iii) we charge $(v,Q)$ to the $1$-level edge border 
$(v_1,Q_1)$. Note that $(v_1,Q_1)$ is charged in this manner
only by $(v,Q)$.

Let us denote by $E_0(\K)$ (resp., $E_1(\K)$) the number of 
$0$-level edge borders (resp., $1$-level edge borders) in $\A(\K)$. 
Then we have the following recurrence:
\begin{equation}\label{Eq:ChargeLevel1}
E_0(\K) \le E_1(\K) + O(n^3) .
\end{equation}
To solve this recurrence, we apply the technique of Tagansky \cite{Tagansky}. Specifically, we remove from $\K$ a randomly chosen set $K\in \K$, and denote by
$\R$ the collection of the $n-1$ remaining sets.  A $0$-level edge 
border $(v,Q)$ in $\A(\K)$, where $v$ is an intersection point of $C_{ij}$ and $C_{k\ell}$ and is 
regular, appears as a $0$-level edge border in $\A(\R)$ if 
and only if $K$ is different from each of the four sets 
$K_i,K_j,K_k,K_\ell$ defining $v$, which happens with probability 
$\frac{n-4}{n}$. A $1$-level edge border $(v,Q)$ in $\A(\K)$ 
becomes a $0$-level edge border in $\A(\R)$ if and only if $K$ is in 
conflict with $(v,Q)$, which happens with probability \emph{at least} 
$\frac{1}{n}$. No other edge border in $\A(\K)$ can
appear as a $0$-level edge border in $\A(\R)$. Hence, we obtain
\begin{equation}\label{Eq:RemoveOneSet}
{\bf E}\left\{E_0(\R)\right\} \geq 
\frac{n-4}{n}E_0(\K) + \frac{1}{n}E_1(\K) ,
\end{equation}
where ${\bf E}$ denotes expectation with respect to the random sample 
$\R$, as constructed above.
Combining (\ref{Eq:ChargeLevel1}) and (\ref{Eq:RemoveOneSet}) yields
$$
\frac{1}{n} E_0(\K) \le \frac{1}{n} E_1(\K) + O(n^2) \le
{\bf E}\left\{E_0(\R)\right\} -
\frac{n-4}{n}E_0(\K) + O(n^2) .
$$
Denoting by $E_0(n)$ the maximum number of $0$-level edge borders
in $\A(\K)$, for any collection $\K$ of size $n$ with the assumed properties, we get the recurrence
$$ 
\frac{n-3}{n} E_0(n) \le E_0(n-1) + O(n^2) ,
$$
whose solution is easily seen to be $E_0(n) = O(n^3\log n)$ (see, e.g., \cite[Proposition 3.1]{Tagansky}).

\subsection{The number of permutation faces}
\label{Subsec:PermutFaces}

Finally, we bound the number of vertices of permutation faces 
using the bound on popular edges just derived. This will also serve as
an upper bound on the number of permutation faces, and thus also on
$g_3(n)$.  We present the analysis in two stages. The first stage
derives the slightly weaker upper bound $O(n^3\log^2n)$, but is 
considerably simpler. The second stage involves a more careful
examination of the possible charging scenarios, and leads to a sharper
recurrence, whose soution is only $O(n^3\log n)$.

Each vertex $v$ is incident to exactly four faces of $\A(\K)$,
so we need to count $v$ with multiplicity of at most $4$---once for 
each permutation face incident to $v$. 
For this we extend the notion
of borders as follows.  The two great circles passing through $v$ 
partition $\SS^2$ into four wedges, or rather slices.  
Each such slice $R$ contains 
a unique face $f$ incident to $v$, and defines, together with $v$, a 
{\em border} $(v,R)$. 
We call $f$ the face \textit{associated} with $(v,R)$.
Similarly to the notation involving edge
borders in Section \ref{Subsec:PopularEdges}, we call $(v,R)$
a {\em popular border}, or a \emph{$0$-level border}, if the face 
associated with $(v,R)$ is a permutation face.  It thus 
suffices to bound the number of $0$-level borders in $\A(\K)$.

If $(v,R)$ is a border in $\A(\K)$ with an associated face $f$, and $\K'$ is a subcollection 
of $\K$, so that $v$ is still a vertex of $\A(\K')$, then $(v,R)$ 
is also a border in $\A(\K')$, except that the face $f'$ of 
$\A(\K')$ associated with $(v,R)$ may be different from $f$ (or, more precisely, properly contain $f$).

If a border $(v,R)$, which is not a $0$-level border in $\A(\K)$, 
becomes a $0$-level border after removing from $\K$ some set $K$, 
we call it a {\em $1$-level border}. The set $K$ is said to be 
\emph{in conflict} with $(v,R)$. Note that $K$ cannot be one 
of the (at most) four sets defining $v$, and that a $1$-level border 
may be in conflict with more than one set of $\K$. 
See Figure \ref{Fig:VertexBorder} (left). 

\begin{figure}[htb]
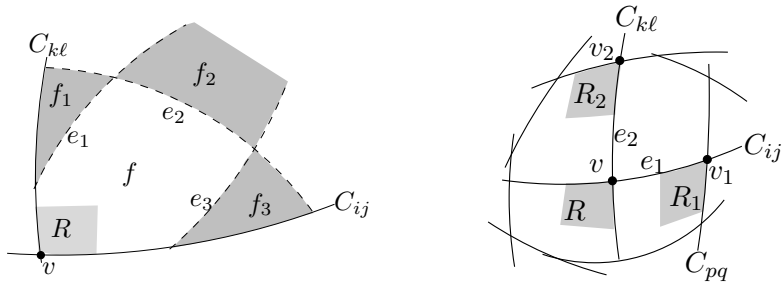

\begin{center}
\input{Weight.pstex_t}\hspace{2cm}\input{VertexBorder.pstex_t}
\caption{\small\sf
Left: A non-permutation face $f$, associated with the $1$-level border $(v,R)$, is separated from permutation faces
$f_1,f_2,f_3$ by the respective edges $e_1\subset
C_{p_1q_1},e_2\subset C_{p_2q_2},e_3\subset C_{p_3q_3}$. If none of
$p_1,q_1,p_2,q_2,p_3,q_3$ belongs to $\{i,j,k,\ell\}$ then 
$(v,R)$ is a 1-level border in conflict with each of 
$K_{p_1},K_{q_1},K_{p_2},K_{q_2},K_{p_3},K_{q_3}$. Right: Charging a $0$-level border $(v,R)$ to the two $1$-level borders
$(v_1,R_1)$, $(v_2,R_2)$, along the two edges $e_1,e_2$ emanating 
from $v$ away from $R$.}
\label{Fig:VertexBorder}
\end{center}
\vspace*{-0.6cm}
\end{figure}

We bound the number of $0$-level borders using a charging scheme 
similar to that in Section~\ref{Subsec:PopularEdges}.
Let $(v,R)$ be a $0$-level border, and let $f$ be the permutation 
face associated with it.  Note that the number of borders incident 
to degenerate vertices is $O(n^3)$.  We may therefore assume that 
$v$ is regular, and let $C_{ij}$ and $C_{k\ell}$ denote the two 
great circles incident to $v$ (so $i,j,k,\ell$ are all distinct). 
Without loss of generality, assume that $R=C_{ij}^+\cap C_{k\ell}^+$.


Let $e_1$ and $e_2$ be the two edges incident to $v$ and emanating
from it away from $R$, where $e_1\subset C_{ij}\cap C_{k\ell}^-$ 
and $e_2\subset C_{k\ell}\cap C_{ij}^-$; 
see Figure \ref{Fig:VertexBorder} (right). Let $v_1$ (resp., $v_2$) be the 
other endpoint of $e_1$ (resp., of $e_2$).  

Our charging scheme is 
based on the following case analysis:

(i) If one of the two edges incident to $v$ and bounding $R$ is 
popular, we charge $(v,R)$ to this edge. Since the number of 
popular edges is $O(n^3\log n)$, and each of them is charged by 
at most four $0$-level borders (twice for each of its endpoints), the 
number of $0$-level borders that fall into this subcase is also 
$O(n^3\log n)$.

(ii) If no edge incident to $v$ and bounding $R$ is popular, we charge 
$(v,R)$ to two $1$-level borders, one incident to $v_1$ and one to 
$v_2$.  Specifically, consider $v_1$, say, and let $C_{pq}$ be the 
circle whose intersection with $C_{ij}$ forms $v_1$, and assume, 
again without loss of generality, that $v$ lies in $C_{pq}^+$. 
We then charge $(v,R)$ to $(v_1,R_1)$, where 
$R_1=C_{pq}^+\cap C_{ij}^+$.  Let $f_1$ be the face of $\A(\K)$ 
associated with $(v_1,R_1)$ (this is the face whose
boundary we trace from $v$ to $v_1$ along $e_1$, and it is 
also incident to $v$).  Since the edge incident to $f,f_1$ 
(and to $v$) is not popular, $f_1$ is not a permutation face.
Clearly, one of the indices $k,\ell$, say $k$, is different from 
both $p,q$. Thus, removing $K_k$ keeps $v_1$ as a vertex in the 
new spherical arrangement, and makes $C_{k\ell}$ disappear, 
so both faces $f,f_1$ fuse into a single larger 
permutation face contained in $R_1$.  Hence, $(v_1,R_1)$ is a 
$1$-level border which is in conflict with $K_k$.
A fully symmmetic argument applies to $v_2$. 
We say that the $1$-level borders $(v_1,R_1)$ and $(v_2,R_2)$, which we charge,
are the \emph{neighbors} of $(v,R)$ in $\A(\K)$.

Note that each $1$-level border $(v',R')$ is charged by at most two 
$0$-level borders in this manner (at most once along each of the two 
edges incident to $v'$ and bounding the face associated with the
border). 

Let $V_0(\K)$ and $V_1(\K)$ denote, respectively, the number of 
$0$-level borders and the number of $1$-level borders in $\A(\K)$
(where we also include degenerate vertices in both counts).
Then we have the following recurrence:
\begin{equation}\label{Eq:VertNaive}
V_0(\K)\leq V_1(\K)+O(n^3\log n) .
\end{equation}
Indeed, each $0$-level border which falls into case (ii) charges
two $1$-level borders, and each $1$-level border is charged at 
most twice. The number of all other $0$-level borders is 
$O(n^3\log n)$, as argued above.
Combining this inequality with the random sampling technique of
Tagansky~\cite{Tagansky}, as in Section~\ref{Subsec:PopularEdges}, 
results in the recurrence
$$
\frac{n-3}{n} V_0(n) \le V_0(n-1) + O(n^2\log n) ,
$$
where $V_0(n)$ is the maximum value of $V_0(\K)$, over all collections
$\K$ of $n$ pairwise disjoint compact convex sets in $\reals^3$.
The solution of this recurrence is $V_0(n) = O(n^3\log^2n)$, which
yields the same upper bound on the number of geometric permutations 
induced by $\K$.

\medskip
\noindent{\bf An improved bound.}
We next improve the bound by replacing the recurrence
(\ref{Eq:VertNaive}) by a refined recurrence. 
Let $(v,R)$ be a $1$-level border which is in conflict with $w\geq 1$ sets 
of $\K$. Then $(v,R)$ becomes a $0$-level border in $\A(K\setminus\{K\})$, 
after removing a random set $K\in \K$, with probability exactly 
$\frac{w}{n}$. Namely, this happens if and only if 
$K$ is one of the $w$ sets in conflict with $(v,R)$. We refer to 
$w$ as the \emph{weight} of $(v,R)$.

In the refined setting, $V_1(\K)$ counts the total weight of all 
the $1$-level borders in $\A(\K)$, so now the contribution of each 
$1$-level border to $V_1(\K)$ is equal to its weight.
By an appropriate adaptation of the argument 
in Section~\ref{Subsec:PopularEdges}, we obtain the following 
{\em equality}:
\begin{equation}\label{Eq:VertexSample}
{\bf E}\{V_0(\R)\}=\frac{n-4}{n}V_0(\K)+\frac{1}{n}V_1(\K),
\end{equation}
where $\R$ denotes a random sample of $n-1$ sets of $\K$.
This follows by noting that the probability of a $1$-level border of
weight $w$ to be counted in $V_0(\R)$ is $\frac{w}{n}$, and it
contributes $w$ to $V_1(\K)$.


In the refined charging scheme, each $1$-level border $(v,R)$ of weight $w\ge 1$ 
gets a supply of $w$ units of charge, which it can give to its 
charging neighboring $0$-level borders. Hence, as long as the number of 
these charging $0$-level borders, which is at most two, does not exceed 
$w$, $(v,R)$ can pay each of its neighbors 1 unit.
Hence, the only problematic case is when $w=1$ and $(v,R)$ is 
charged twice. The following technical lemma takes care of this
case.
\begin{lemma}\label{Thm:Refine}
The number of $1$-level borders having weight $1$ and charged by two 
$0$-level borders is $O(n^3\log n)$.
\end{lemma}
Before proving Lemma~\ref{Thm:Refine}, we show how to use it 
to replace \ref{Eq:VertNaive} by a better recurrence, and thereby establish an improved bound on the number of geometric 
permutations in $\reals^3$. 

If a $1$-level border $(v,R)$ has only one neighboring $0$-level 
border $(v',R')$ then $(v',R')$ can receive one unit of charge from 
$(v,R)$, regardless of what the weight of $(v,R)$ is.  Similarly, 
if $(v,R)$ has weight at least $2$, and it has two neighboring 
$0$-level borders, each of these $0$-level borders can receive one 
unit of charge from $(v,R)$.  The number of remaining $1$-level 
borders, namely the $1$-level borders of weight $1$ with two neighboring 
$0$-level borders, is $O(n^3\log n)$, by Lemma~\ref{Thm:Refine}.

To recap, each $0$-level border, except possibly for $O(n^3\log n)$ ones, 
receives $1$ unit of charge from each of its two neighboring 
$1$-level borders. Moreover, the number of charges made to each of 
the remaining $1$-level borders, by its neighboring $0$-level borders,
does not exceed its weight.
Thus, we can replace (\ref{Eq:VertNaive}) by the following inequality:
$$
2V_0(\K)\leq V_1(\K)+O(n^3\log n).
$$
Combining this with (\ref{Eq:VertexSample})
 we get
$$
\frac{2}{n}V_0(\K) \le \frac{1}{n}V_1(\K) 
 + O(n^2\log n) \le
{\bf E}\left\{V_0(\R)\right\} - \frac{n-4}{n}V_0(\K) + O(n^2\log n) ,
$$
or
$$
\frac{n-2}{n}V_0(\K) \le 
{\bf E}\left\{V_0(\R)\right\} + O(n^2\log n) .
$$
Replacing $V_0(\K)$, $V_0(\R)$ by their respective maximum values
$V_0(n)$, $V_0(n-1)$, we thus obtain the recurrence
$$
\frac{n-2}{n}V_0(n) \le V_0(n-1) + O(n^2\log n) ,
$$
whose solution is easily seen to be $V_0(n) = O(n^3\log n)$ (again, see \cite[Proposition 3.1]{Tagansky}).

As mentioned earlier, $V_0(n)$ serves as an upper bound on the 
number of geometric permutations induced by $\K$.
We thus conclude with the following main result of this section.
\begin{theorem}\label{Thm:Main3D}
Any collection $\K$ of $n$ pairwise disjoint convex sets in 
$\reals^3$ admits at most $O(n^3\log n)$ geometric permutations.
\end{theorem}

\paragraph{Proof of Lemma \ref{Thm:Refine}.}
Consider a $1$-level border $(v,R)$ of weight $1$, where
$v$ is incident to two great circles $C_{ij}$, $C_{k\ell}$, which 
is charged twice.  We may assume that $v$ is regular (i.e., the 
four indices $i,j,k,\ell$ are distinct), since the number of 
remaining borders is $O(n^3)$.
Let $(v_1,R_1)$ be the $0$-level border that 
charges $(v,R)$ along $C_{ij}$, and let $(v_2,R_2)$ be the 
$0$-level border that charges $(v,R)$ along $C_{k\ell}$.  
By construction, both $v_1$ and $v_2$ are regular 
(otherwise they do not charge $v$).
Let $C_{p_1q_1}$ denote the other circle incident 
to $v_1$, and let $C_{p_2q_2}$ denote the other circle incident 
to $v_2$.  Clearly, each index in $\{p_1,q_1,p_2,q_2\}$
which does not belong to $\{i,j,k,\ell\}$ contributes to the 
weight of $(v,R)$, so, by assumption, there is only one such 
index, call it $q$. Since $v_1$ is regular, neither $p_1$ nor 
$q_1$ belongs to $\{i,j\}$, so (exactly) one of them must belong 
to $\{k,\ell\}$, say $p_1=k$ and then $q_1=q$. Symmetrically, we 
may assume that $p_2=i$, say, and then $q_2=q$. 
Since $v$ is regular and $q\notin \{i,j,k,\ell\}$, the two circles $C_{p_1q_1}$, 
$C_{p_2q_2}$ (i.e., $C_{kq}$, $C_{iq}$) are distinct.
See Figure~\ref{weight1}.

\begin{figure}[htb]
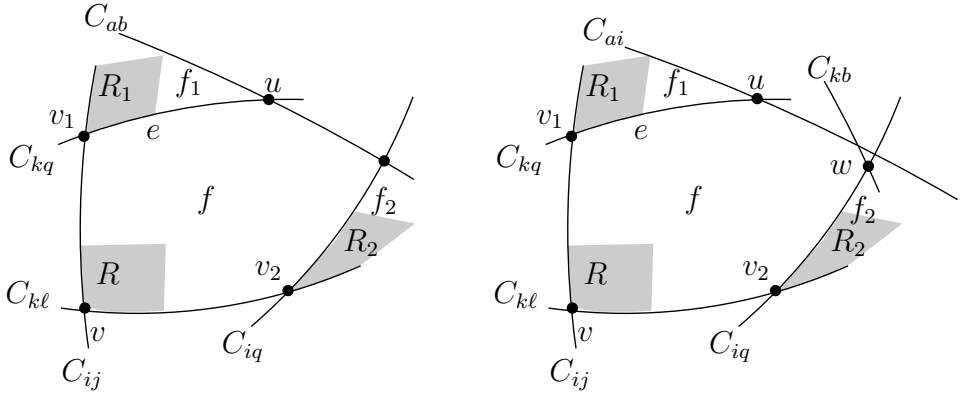

\begin{center}
\input{RefineWeight.pstex_t}\hspace{1cm}\input{RefineWeight2.pstex_t}
\caption{\small\sf
Two scenarios depicting a $1$-level border $(v,R)$ of weight $1$ 
that is charged by two $0$-level borders $(v_1,R_1),(v_2,R_2)$.} 
\label{weight1}
\end{center}
\vspace*{-0.6cm}
\end{figure}

In this special scenario we have two distinct permutation faces 
$f_1$ and $f_2$, where $f_1$ is the face associated with 
$(v_1,R_1)$ and $f_2$ is the face associated with $(v_2,R_2)$.

There are two possible subcases: Assume first that the face $f$
associated with $(v,R)$ is just the quadrangle bounded by
$C_{ij}$, $C_{k\ell}$, $C_{iq}$ and $C_{kq}$. In this case the
fourth vertex of $f$, formed by intersection of $C_{iq}$ and 
$C_{kq}$, is degenerate. Since each degenerate vertex is incident 
to at most four faces, the number of $1$-level borders 
falling into this subcase is $O(n^3)$.

Suppose then that $f$ has additional edges and vertices.
Consider, for example, the vertex $u$ which is the other endpoint
(other than $v_1$) of the edge $e$ of $f$ lying on $C_{kq}$. 
Let $C_{ab}$ denote the other circle incident to $u$. 
Assume with no loss of generality that $v$ lies in the hemisphere 
$C_{ab}^+$. We may also assume that neither $a$ nor $b$ is in 
$\{k,q\}$, for otherwise $u$ is a degenerate vertex, so we can 
argue similarly to the previous subcase.

Suppose first that neither $a$ nor $b$ is equal to $i$. Then
removing $K_i$ keeps $u$ as a vertex of $\A(\K\setminus\{K_i\})$. The
edge $e$ extends at its other end into a longer {\em popular}
edge (it bounds on one side an extension of $f\cup f_2$ and on the other
side an extension of $f_1$, both of which are now permutation faces; 
see Figure~\ref{weight1} (left)), 
so $(u,C_{ab}^+)$ is a $1$-level edge border. We charge the 
1-level border $(v,R)$ to $(u,C_{ab}^+)$. By construction, such an
edge border is charged only once, as is easily checked.

The number of $1$-level edge borders can be bounded using the
Clarkson-Shor analysis technique \cite{CS}, similar to the way it was
used in the proof of Theorem~\ref{Thm:NZero}. That is, since each
$1$-level edge border is defined by at most four sets of $\K$ and becomes
a $0$-level edge border when we remove (at least) one set from $\K$,
the number of $1$-level 
edge borders is $O\left({\bf E}\{E_0(\K')\}\right)$, where $\K'$ 
is a random sample of $n/2$ sets of $\K$. Hence, the analysis in the
preceding subsection implies that the number of $1$-level edge 
borders in $\A(\K)$ is $O(n^3\log n)$, and therefore the same bound
holds for the number of $1$-level borders $(v,R)$ under consideration.

We are therefore left with the situation where, say, $b=i$.
Applying a fully symmetric argument to the edge of $f$ lying on
$C_{iq}$, we conclude that the only problematic case is where $f$ is
at least pentagonal, with five consecutive vertices $u,v_1,v,v_2,w$,
so that $u$ is incident to $C_{ai}$ and $C_{kq}$,
$v_1$ is incident to $C_{kq}$ and $C_{ij}$,
$v$ is incident to $C_{ij}$ and $C_{k\ell}$,
$v_2$ is incident to $C_{k\ell}$ and $C_{iq}$, and
$w$ is incident to $C_{iq}$ and $C_{kb}$; here $a$ and $b$ are two
indices, neither of which belongs to $\{i,j,k,\ell,q\}$;
$a$ and $b$ may be equal. See Figure~\ref{weight1} (right).

Let $\A_i$ be the arrangement of the $n-1$ great circles of the form 
$C_{ir}$ or $C_{ri}$, for $r\neq i$.
Let $f_0$ be the face of $\A_i$ containing $f$.
By assumption, the boundary of $f$ touches three distinct boundary 
edges of $f_0$.  We charge the 1-level border $(v,R)$ to the triple
$(f_0,e_0,e_1)$, where $e_0\subset C_{ij}$ and $e_1\subset C_{iq}$
are the two boundary edges of $f_0$ which contain the respective edges
of $\partial f$. To complete the proof of Lemma \ref{Thm:Refine}, 
we need the following two lemmas. 

\begin{lemma} \label{nok33}
Let $1\leq i\leq n$, and let $\A_i$ be the arrangement of the
$n-1$ great circles $C_{ir}$ or $C_{ri}$, for $r\neq i$. Let $f_0$ 
be a face in $\A_i$, and let $e_0,e_1$ be two edges of $f_0$. Then
there exist at most two faces of $\A$ which are contained in $f_0$ and
are bounded by a portion of $e_0$, by a portion of $e_1$, and by a 
portion of some other edge of $f_0$.
\end{lemma}
\begin{proof}
The edges $e_0$ and $e_1$ partition $\bd f_0$ into up to four connected
portions, $e_0$, $\gamma^-$, $e_1$, $\gamma^+$. We claim that there can
be at most one face $f$ of $\A$ which is contained in $f_0$ and 
which is bounded by a portion of $e_0$, a portion of $e_1$, and a 
portion of $\gamma^+$. A symmetric claim holds if we replace 
$\gamma^+$ by $\gamma^-$, and the lemma follows. The latter claim
follows by observing that the existence of two distinct faces
$f_1,f_2$ of $\A$ contained in $f_0$ and touching $e_0$, $e_1$ and
$\gamma^+$ would lead to an impossible planar drawing of $K_{3,3}$, 
as illustrated in Figure~\ref{k33}. See, e.g., \cite{EzS} for a 
similar argument.
\end{proof}

\begin{figure}[htb]
\begin{center}
\input{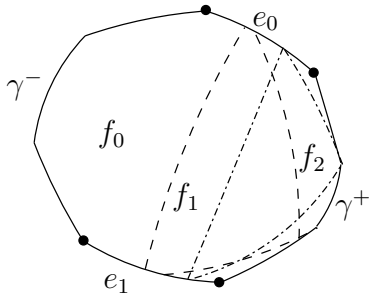}
\caption{\small\sf
A face $f_0$ of $\A_i$ cannot contain two distinct faces 
$f_1,f_2$ of $\A(\K)$ that touch $e_0,e_1$ and $\gamma^+$.}
\label{k33}
\end{center}
\vspace*{-0.6cm}
\end{figure}

\begin{lemma} \label{zonesq}
The number of triples $(f_0,e_0,e_1)$, where $f_0$ is a face in 
$\A_i$, as defined in Lemma \ref{nok33}, and $e_0,e_1$ are two 
edges of $f_0$, summed over all $i$, is $O(n^3)$.
\end{lemma}
\begin{proof}
This follows from the well known result that the sum of the squares of
the face complexities in an arrangement of $n$ lines in the plane is
$O(n^2)$; see, e.g., \cite{SA}. 
The same analysis applies to an arrangement of great circles
on the unit sphere. Summing this bound over all $i$, the lemma follows.
\end{proof}

Lemma \ref{nok33} implies that any triple $(f_0,e_0,e_1)$, as above, 
is charged by at most four 1-level borders $(v,R)$.
Indeed, the triple determines at most two possible faces $f$ 
of $\A$, and the edge $e_0$ determines a 
unique edge of $f$ with $v$ as one of its endpoints.
By Lemma~\ref{zonesq}, the overall number of charged triples
$(f_0,e_0,e_1)$ is $O(n^3)$, so the overall number of 
$1$-level borders $(v,R)$ falling into the last subcase is $O(n^3)$.
This completes the proof of Lemma~\ref{Thm:Refine}.\hfill $\square$


%



\section{Geometric Permutations in Higher Dimensions}\label{Sec:Higher}
In this section we generalize Theorem~\ref{Thm:Main3D} by 
showing that the number of geometric permutations induced by 
any collection $\K=\{K_1,\ldots,K_n\}$ of $n$ pairwise disjoint 
convex sets in $\reals^d$ is 
$O(n^{2d-3}\log n)$, for any $d\geq 3$.

\medskip
\noindent{\bf Setup.}
The basic setup is similar to that in three dimensions, but we repeat
it here for the sake of readability. Specifically, we may assume, using the same reasoning as
before, that the sets of $\K$ are compact (in addition to being 
pairwise disjoint and convex).
For each $1\leq i<j\leq n$ we fix some hyperplane $h_{ij}$ which strictly
separates $K_i$ and $K_j$.  We orient $h_{ij}$ so that $K_i$ lies 
in the negative open halfspace $h_{ij}^-$ that it bounds, and 
$K_j$ lies in the positive open halfspace $h_{ij}^+$.
We represent directions of lines in $\reals^d$ by points on the 
unit $(d-1)$-sphere $\SS^{d-1}$. We may assume that the
separating hyperplanes $h_{ij}$ are in {\em general position}, so that
every $d$ of them intersect in a unique point, and no $d+1$ of them
have a point in common. 

Each separating hyperplane $h_{ij}$ induces a great $(d-2)$-sphere 
$C_{ij}$ on $\SS^{d-1}$, which is the locus of the directions of 
all lines parallel to $h_{ij}$.
$C_{ij}$ partitions $\SS^{d-1}$ into two open hemispheres 
$C_{ij}^+$, $C_{ij}^-$, so that $C_{ij}^+$ (resp., $C_{ij}^-$), 
consists of the directions of lines which cross $h_{ij}$ from 
$h_{ij}^-$ to $h_{ij}^+$ (resp., from $h_{ij}^+$ to $h_{ij}^-$). 
Any oriented common transversal line of $K_i$ and $K_j$ visits 
$K_j$ after (resp., before) $K_i$ if and only if its direction 
lies in $C_{ij}^+$ (resp., in $C_{ij}^-$).

Put $\C(\K)=\{C_{ij}\mid 1\leq i<j\leq n\}$, and
consider the arrangement $\A(\K)$ of these $n\choose 2$ 
$(d-2)$-spheres on $\SS^{d-1}$. 
It partitions $\SS^{d-1}$ into relatively open cells of dimensions $0,1,\ldots,d-1$; we refer to an $s$-dimensional cell of $\A(\K)$ simply as an $s$-cell.
The assumption that the hyperplanes 
$h_{ij}$ are in general position implies that the $(d-2)$-spheres 
of $\C(\K)$ are also in general position, in the sense that the 
intersection of any $s$ distinct spheres of $\C(\K)$, for 
$1\leq s\leq d-1$, is a $(d-s-1)$-sphere, and the intersection 
of any $d$ distinct spheres of $\C(\K)$ is empty.  
Each $(d-1)$-cell $f$ of $\A(\K)$ induces a 
relation $\prec_f$ on $\K$, in which $K_i\prec_f K_j$ 
(resp., $K_j\prec_f K_i$) if $f\subseteq C_{ij}^+$ 
(resp., $f\subseteq C_{ij}^-$).
The direction of each oriented line transversal $\lambda$ 
of $\K$ belongs to the unique $(d-1)$-cell $f$ of $\A(\K)$ whose relation 
$\prec_f$ coincides with the linear order in which $\lambda$ visits 
the sets of $\K$. In particular, as noted by Wenger \cite{Wen2}, the number 
of geometric permutations is bounded by the number of $(d-1)$-cells 
of $\A(\K)$, which is $O(n^{2d-2})$.

We call a $(d-1)$-cell $f$ of $\A(\K)$ a \emph{permutation cell} if there 
is at least one line transversal of $\K$ whose direction belongs 
to $f$.  As in the three-dimensional case, we improve the above bound by showing that the number of
permutation cells in $\A(\K)$ is $O(n^{2d-3}\log n)$, which also 
bounds the number of geometric permutations induced by $\K$.

We refer to $0$-cells in $\A(\K)$ as \emph{vertices}, and 
to $1$-cells as \emph{edges}.  We say that a vertex $v$ of $\A(\K)$ 
is \emph{regular} if the $d-1$ $(d-2)$-spheres of $\C(\K)$ that are
incident to $v$ are defined by $2d-2$ distinct sets of $\K$; otherwise
$v$ is a {\em degenerate} vertex.  Clearly, the number of degenerate 
vertices is $O(n^{2d-3})$, so it suffices 
to bound the number of regular vertices of permutation cells.

As in the three-dimensional case, we will also consider subcollections 
$\K'$ of $\K$, typically obtained by removing one set, say $K_q$, 
from $\K$. Doing so eliminates all separating hyperplanes $h_{iq}$, 
$h_{qi}$, as well as all the corresponding $(d-2)$-spheres $C_{iq}$, 
$C_{qi}$, and $\A(\K')$ is constructed only from the remaining spheres. 
In particular, a vertex\footnote{%
  As in the three-dimensional case, the intersection consists of two
  antipodal points, so there are two choices for $v$.}
$v$ of the intersection 
$C_{i_1j_1}\cap C_{i_2j_2}\cap\cdots\cap C_{i_{d-1}j_{d-1}}$ 
of $\A(\K)$ remains a vertex of $\A(\K')$ if and only if 
$q\not \in \{i_1,j_1,\ldots, i_{d-1},j_{d-1}\}$. A cell of 
$\A(\K')$, of any dimension $s\ge 1$, may contain several cells 
of $\A(\K)$. As before, if $f'$ is a $(d-1)$-cell of $\A(\K')$ 
which contains a permutation cell $f$ of $\A(\K)$ then $f'$ is 
a permutation cell in $\A(\K')$; the permutation that it induces 
is the permutation of $f$ with $K_q$ removed.

Each $s$-cell $f$ of $\A(\K)$ is incident to $2^{d-s-1}$ 
$(d-1)$-cells of $\A(\K)$.  If all these cells are permutation 
cells, $f$ is called \emph{popular}.
In particular, a popular vertex is incident to $2^{d-1}$ permutation
cells, a popular edge is incident to $2^{d-2}$ permutation cells, and
a popular $(d-1)$-cell is a permutation cell.

\medskip
\noindent{\bf Overview.} 
We show that the number of popular vertices is $O(n^{2d-3})$ by a straightforward 
generalization of the analysis in Section~\ref{subsec:pop3D}.
The analysis then proceeds by applying, for each $1\leq s\leq d-1$,
a charging scheme, which expresses the number of popular $s$-cells 
in terms of the number of popular $(s-1)$-cells (and degenerate 
vertices). A naive charging 
scheme produces a recurrence whose solution incurs an additional
logarithmic factor for each $s$, resulting in the weaker bound 
$g_d(n) = O(n^{2d-3}\log^{d-1}n)$. A more careful analysis, as 
in the three-dimensional case, leads to refined recurrences, whose
solution yields the improved bound $g_d(n) = O(n^{2d-3}\log n)$. 
(We lose a logarithmic factor only when passing from vertices to
edges, as in the three-dimensional case.)

\subsection{The number of popular vertices}

For a regular vertex $v\in\bigcap_{q=1}^{d-1}C_{i_qj_q}$ of $\A(\K)$, we 
denote by $\K_v$ the collection $\{K_{i_q},K_{j_q}\mid 1\leq q\leq d-1\}$ 
of the $2d-2$ sets defining $v$.

\begin{lemma}\label{Thm:ConsecutivePairsHigh}
Let $v\in\bigcap_{q=1}^{d-1}C_{i_qj_q}$ be a regular popular vertex of
$\A(\K)$. \\
(i) Each pair of sets $K_a\in \K_v$ and $K_b\in \K\setminus\K_v$
appear in the same order in all the $2^{d-1}$ 
permutations induced by the $(d-1)$-cells incident to $v$.\\
(ii) The elements of each pair $K_{i_q},K_{j_q}\in \K_v$, for 
$1\leq q\leq d-1$, are consecutive in all these $2^{d-1}$ 
permutations.
\end{lemma}  
\begin{proof}
Each pair of distinct $(d-1)$-cells $f,g$ incident to $v$ are 
separated by at most $d-1$ $(d-2)$-spheres from 
$\{C_{i_1j_1},\ldots,C_{i_{d-1}j_{d-1}}\}$, and only by these spheres.
Hence the orders $\prec_f$ and $\prec_g$ may disagree only over the 
pairs $(K_{i_q},K_{j_q})$, for $1\leq q\leq d-1$.
As in the proof of Lemma~\ref{Thm:ConsecutivePairs}, this is 
easily seen to imply both parts of the lemma.
\end{proof}

\begin{lemma}\label{Thm:PopularLineHigh}
Let $v\in\bigcap_{q=1}^{d-1}C_{i_qj_q}$ be a regular popular vertex in
$\A(\K)$. Then the line $\lambda_v=\bigcap_{q=1}^{d-1}h_{i_qj_q}$ 
stabs all the $n-2d+2$ sets in $\K\setminus\K_v$, and misses all 
the $2d-2$ sets in $\K_v$.
\end{lemma}
\begin{proof}
By definition, $\lambda_v$ misses every set $K\in \K_v$, because 
it is contained in a hyperplane separating $K$ from another set in 
$\K_v$.  Hence, it suffices to show that $\lambda_v$ is a 
transversal of $\K\setminus\K_v$.

To show this, we fix a set $K_a\in\K\setminus\K_v$ and show that
each of the $2^{d-1}$ wedges determined by 
$\{h_{i_qj_q}\mid 1\leq q\leq d-1\}$ meets $K_a$. Each of these wedges is the intersection of $d-1$ halfspaces, where the $q$-th halfspace is either $h_{i_qj_q}^+$ or $h_{i_qj_q}^-$, for $q=1,\ldots,d-1$. All these wedges have $\lambda_q$ on their boundary, and the convexity 
of $K_a$ then implies, exactly as in the three-dimensional case,
that $\lambda_v$ intersects $K_a$.

For specificity, we show that $K_a$ intersects the wedge
$\bigcap_{q=1}^{d-1} h_{i_qj_q}^+$; 
the proof for the other wedges is essentially the same.
Lemma~\ref{Thm:ConsecutivePairsHigh} implies that $K_a$ lies at the 
same position in each of the $2^{d-1}$ permutations induced by the 
cells incident to $v$. 
For each index $q$, if $K_{i_q},K_{j_q}$ appear before $K_a$ 
(resp., after $K_a$) in all permutations induced by the cells 
incident to $v$, put $C_q = C_{i_qj_q}^+$ 
(resp., $C_q = C_{i_qj_q}^-$).

Let $f$ be the cell incident to $v$ and contained in
$\bigcap_{q=1}^{d-1} C_q$, and let $\mu_f$ be a transversal 
line stabbing $\K$ in the order $\prec_f$ (so its direction lies in
$f$).  By the choice of $f$ and by our assumption, we have either
$K_{i_q}\prec_f K_{j_q} \prec_f K_a$, or
$K_{a}\prec_f K_{j_q} \prec_f K_{i_q}$. This implies in the former
case that $\mu_f$ visits $K_a$ after crossing $h_{i_qj_q}$ from 
$h_{i_qj_q}^-$ (the side containing $K_{i_q}$) to $h_{i_qj_q}^+$ 
(the side containing $K_{j_q}$). In the latter case, $\mu_f$ first
visits $K_a$ and then crosses $h_{i_qj_q}$ from
$h_{i_qj_q}^+$ to $h_{i_qj_q}^-$.  Thus, in either case, the segment 
$\lambda_f\cap K_a$ lies in $h_{i_qj_q}^+$, and this holds for 
every $1\le q\le d-1$. Hence
$\lambda_f\cap K_a \subset \bigcap_{q=1}^{d-1} h_{i_qj_q}^+$,
and the claim follows. 
\end{proof}

\begin{theorem}\label{Thm:NZeroHigh}
Let $\K$ be a collection of $n$ pairwise disjoint compact convex sets 
in $\reals^d$. Then the number of popular vertices in $\A(\K)$ is 
$O(n^{2d-3})$. 
\end{theorem}
\begin{proof}
As in the three-dimensional case, it is easily checked that a popular vertex must be regular.
Let $v\in\bigcap_{q=1}^{d-1}C_{i_qj_q}$ be a (regular) popular vertex in 
$\A(\K)$, and let $\lambda_v=\bigcap_{q=1}^{d-1}h_{i_qj_q}$ be the 
intersection line of the corresponding hyperplanes. Consider the plane 
$H=\bigcap_{q=1}^{d-2}h_{i_qj_q}$, put $K^*_a = K_a\cap H$, for 
each index $a\not\in \{i_q,j_q\mid 1\leq q\leq d-2\}$, and 
denote by $\K^*$ the collection of these $n-2d+4$ planar cross-sections.
Clearly, all sets in $\K^*$ are pairwise disjoint, compact, and
convex.

By Lemma~\ref{Thm:PopularLineHigh}, $\lambda_v$ lies in $H$, stabs all 
the sets in $\K^*\setminus\{K^*_{i_{d-1}},K^*_{j_{d-1}}\}$, and misses 
the two sets $K^*_{i_{d-1}},K^*_{j_{d-1}}$.  As in 
Theorem \ref{Thm:NZero}, we charge $\lambda_v$ to an extremal line 
$\mu=\mu_v$ within $H$ which is tangent to two sets of $\K^*$, and 
misses only the sets among $K^*_{i_{d-1}},K^*_{j_{d-1}}$ that it does 
not touch. As in the preceding analysis, each extremal line $\mu$ 
of this kind is charged at most twice.
Applying the Clarkson-Shor analysis \cite{CS}, similarly to 
Theorem~\ref{Thm:NZero}, the number of lines $\mu$, 
charged within $H$, is $O(n)$. Summing over all possible choices 
of the 2-planes $H$, namely over all choices of $d-2$ of the 
hyperplanes $h_{ij}$, the number of lines $\lambda_v$, and thus the
number of popular vertices, is $O(n\cdot n^{2d-4})=O(n^{2d-3})$.
\end{proof}

\subsection{The number of permutation cells}

We next generalize the analysis of Section~\ref{Subsec:PermutFaces} 
to higher dimensions. We first extend the notion of borders. Let 
$v$ be a vertex of $\A(\K)$, so that
$v\in\bigcap_{1\leq q\leq d-1}C_{i_qj_q}$.  For any subset $J$ 
of $\{1,\ldots,d-1\}$, let $R\subseteq\SS^{d-1}$ be a connected 
component of $\SS^{d-1}\setminus \bigcup_{q\in J} C_{i_qj_q}$. 
Equivalently, it is the intersection of $|J|$ hemispheres, where
the $q$-th hemisphere, for $q\in J$, is either
$C_{i_qj_q}^+$ or $C_{i_qj_q}^-$. Note that there are $2^{|J|}$ 
such regions (for any fixed $J$). 
We call $(v,R)$ an {\em $s$-border}, where $s=|J|$.
Given $v$ and $R$, for $s\ge 1$, there is a unique 
$s$-dimensional cell $f$ of $\A(\K)$ which is incident to $v$ and 
is contained in the interior of $R$. This cell $f$ lies in the intersection of the interior of $R$ with $\bigcap_{q\in J^c}C_{i_qj_q}$, where $J^c=\{1,\ldots,d-1\}\setminus J$.
We refer to $f$ as the $s$-cell of
$\A(\K)$ associated with $(v,R)$. For $s=0$ we define the  $s$-cell of
$\A(\K)$ associated with $(v,R)$ to be $v$ itself, and for $s=d-1$
we define the  $s$-cell of
$\A(\K)$ associated with $(v,R)$ to be  the unique $(d-1)$-cell incident to $v$ and contained in $R$.
 The reader is invited to check
that, for $d=3$, a $0$-border, in the new definition, is a vertex 
of $\A(\K)$, a $1$-border is an edge border, and a $2$-border is what
we simply called a border.

Let $(v,R)$ be an $s$-border and let $f$ be the $s$-cell associated
with $(v,R)$. If $f$ is a popular cell, we say that $(v,R)$ is a 
\emph{$0$-level} \emph{$s$-border} of $\A(\K)$.  An $s$-border $(v,R)$ 
is a \emph{$1$-level} \emph{$s$-border} in $\A(\K)$ if it is not 
a $0$-level $s$-border, but becomes such a border after removing
from $\K$ some single set $K$. In this case we say that $K$ is 
\emph{in conflict} with $(v,R)$. As in the three-dimensional case, 
$K$ need not be unique.

For each $t = 0,1$ and $0\leq s\leq d-1$, let $N_t^{(s)}(\K)$ 
be the number of $t$-level $s$-borders in $\A(\K)$, and let
$N_t^{(s)}(n)$ denote the maximum value of $N_t^{(s)}(\K)$, 
over all collections $\K$ of $n$ pairwise disjoint compact 
convex sets in $\reals^d$.

Note that $N_0^{(0)}(\K)$ is the number of popular 
vertices in $\A(\K)$, so we have $N_0^{(0)}(n) = O(n^{2d-3})$.
The term $N_0^{(d-1)}(\K)$ counts the overall number of 
vertices incident to permutation cells, where each vertex is 
counted once for each permutation cell incident to it.
Assuming $n\geq d$, each permutation cell in $\A(\K)$ is incident to
at least one vertex (and each vertex is incident to at most $2^{d-1}$ permutation cells). Thus, the number of geometric permutations 
of $\K$ is at most $N_0^{(d-1)}(\K)$. 

For each $1\leq s\leq d-1$, we apply a charging scheme, which results
in a recurrence which expresses $N_0^{(s)}(\K)$ in terms of 
$N_0^{(s-1)}(\K)$ and $N_1^{(s)}(\K)$.

Fix $1\leq s\leq d-1$. Let $(v,R)$ be a $0$-level $s$-border 
in $\A(\K)$, and let $f$ be the popular $s$-cell associated with
$(v,R)$.  
Let $C_{i_1j_1},C_{i_2j_2},\ldots, C_{i_{d-1}j_{d-1}}$ 
be the $d-1$ $(d-2)$-spheres of $\C(\K)$ incident to $v$, and assume,
with no loss of generality, that $R=\bigcap_{q=1}^{s} C_{i_qj_q}^+$. 
Moreover, we may assume that $v$ is regular, since the number of 
$s$-borders incident to degenerate vertices is clearly $O(n^{2d-3})$.
To simplify the notation, we refer to borders incident to a regular
(resp., degenerate) vertex as {\em regular borders} (resp.,
{\em degenerate borders}).

For each $1\leq q\leq s$ there exists a unique edge $e_q^+$ of $f$
which is incident to $v$ and not contained in $C_{i_qj_q}$. Indeed,
by construction, $f$ lies in the intersection $s$-sphere 
$\bigcap_{q=s+1}^{d-1} C_{i_qj_q}$ (for $s=d-1$, this is the 
entire $\SS^{d-1}$), 
and each edge of $f$ incident to $v$ is formed by further
intersecting this sphere with $s-1$ additional spheres from
$C_{i_1j_1},\ldots,C_{i_sj_s}$.  The claim follows since only one 
side of the resulting intersection circle lies (near $v$) in the 
closure of $R$.
Let $e_q^-$ denote the other edge of $\A(\K)$ which is incident 
to $v$ and lies on the same intersection circle $\gamma$ as $e_q^+$, 
so $e_q^-$ emanates from $v$ away from $R$. Let $v_q$ denote the
other endpoint of $e_q^-$, and let $g_q$ denote the (unique) $(s-1)$-cell 
which bounds $f$, lies in
$C_{i_qj_q}\cap\left(\bigcap_{t=s+1}^{d-1} C_{i_tj_t}\right)$, 
is incident to $v$ and is contained in 
$R_q = \bigcap_{1\le t\le s,\,t\ne q} C_{i_tj_t}^+$. 
Also, $g_q$ is the $(s-1)$-cell associated with $(v,R_q)$.
See Figure \ref{Fig:ChargeHigh}. 
There are two possible cases:

\begin{figure}[htbp]
\begin{center}
\input{BordersHigh.pstex_t}
\caption{\small\sf
Charging a $0$-level $s$-border along the edge $e_q$.} 
\label{Fig:ChargeHigh}
\end{center}
\end{figure}

(i) If one of the $(s-1)$-cells $g_1,g_2,\ldots,g_s$, say $g_s$, is
popular, we charge $(v,R)$ to the $0$-level $(s-1)$-border $(v,R_s)$, 
noting, as above, that $g_s$ is the $(s-1)$-cell associated with
this border.  By construction, each $0$-level $(s-1)$-border $(v,R)$ is charged at
most $2(d-s)$ times in this manner, once from each $s$-border
associated with an $s$-cell which is bounded by the $(s-1)$-cell 
associated with this border (there are $d-s$ choices for the great sphere $C_{i_tj_t}$ that participates in the definition of $(v,R)$ but is absent in the $s$-border, and two choices of the corresponding hemisphere $C_{i_tj_t}^+,C_{i_tj_t}^-$).
Hence, the number of $0$-level $s$-borders falling into subcase (i) is 
$O\left(N_0^{(s-1)}(\K)\right)$.

(ii) None of the $(s-1)$-cells $g_1,g_2,\ldots,g_s$ is popular.
For each $1\leq q\leq s$, let $C_{k_q\ell_q}$ be the additional
great sphere incident to $v_q$, and suppose, for specificity,
that $v\in C_{k_q\ell_q}^+$. The vertex $v_q$ participates in the
$1$-level $s$-border $(v_q,R'_q)$, where 
$R'_q = C_{k_q\ell_q}^+\cap
\left(\bigcap_{1\le t\le s,\,t\ne q} C_{i_tj_t}^+\right)$. 

Since $g_q$ is not popular, $(v_q,R'_q)$ is not a $0$-level
$s$-border. Let $f_q$ be the $s$-cell associated with $(v_q,R'_q)$. 
Clearly, at least one of $i_q,j_q$ does not belong to $\{k_q,\ell_q\}$; 
say it is $i_q$.  Thus, and since $v$ is regular, removing $K_{i_q}$ keeps $v_q$ (and hence 
$(v_q,R'_q)$) intact, and makes $f$ and $f_q$ fuse into a larger 
$s$-cell $f'$ containing both of them. Clearly, $f'$ is the cell 
associated with $(v_q,R'_q)$ in $\A(\K\setminus\{K_{i_q}\})$,
and it is popular there because $f\subset f'$ was popular in
$\A(\K)$. We say that the borders $(v,R)$, $(v_q,R'_q)$ are 
{\em neighbors} in $\A(\K)$.

We then charge $(v,R)$ to its $s$ neighboring $1$-level $s$-borders 
$(v_q,R'_q)$, for $q=1,\ldots,s$. 
Note that each $1$-level $s$-border $(v,R)$ is charged at most $s$ 
times, once along each of the $s$ edges, incident to $v$, of the $s$-cell associated with it.
We thus obtain the following recurrence.
\begin{equation}\label{Eq:ChargeNaiveHigh}
N_0^{(s)}(\K)\leq N_1^{(s)}(\K) + 
O\left(N_0^{(s-1)}(\K)+n^{2d-3}\right),
\end{equation}
where the first term in the right hand side bounds the number of 
$0$-level $s$-borders falling into case (ii), and the second term 
bounds the number of the remaining $0$-level $s$-borders.

Similarly to the three-dimensional case, we combine the system 
(\ref{Eq:ChargeNaiveHigh}) of recurrences with the analysis technique
of Tagansky, and solve the resulting recurrences to obtain a slightly
inferior bound (involving a larger polylogarithmic factor). We then
refine the recurrences, using a more careful analysis, similar to 
the one in Section~\ref{Sec:Perms3D}, and thereby obtain the improved
bound $O(n^{2d-3}\log n)$.

\medskip
\noindent{\bf Applying Tagansky's technique: The simpler variant.} We prove 
that $N_0^{(s)}(n)=O(n^{2d-3}\log^s n)$ by induction on $s$.
For the base case $s=0$, we have $N_0^{(0)}(n)=O(n^{2d-3})$ by
Theorem~\ref{Thm:NZeroHigh}. Consider a fixed $s\ge 1$ and assume
that the bound holds for $s-1$, so (\ref{Eq:ChargeNaiveHigh}) becomes  
\begin{equation}\label{Eq:ChargeEdgesHigh}
N_0^{(s)}(\K)\leq N_1^{(s)}(\K)+O(n^{2d-3}\log^{s-1}n).
\end{equation}

Let $\R$ be a random sample of $n-1$ sets of $\K$, obtained by 
removing a random set $K$ from $\K$. The expected number of 
$0$-level popular $s$-borders in $\A(\R)$ satisfies 
\begin{equation}\label{Eq:HighCS}
{\bf E}\{N_0^{(s)}(\R)\}\geq \frac{n-2d+2}{n}N_0^{(s)}(\K)+\frac{1}{n}N_1^{(s)}(\K) .
\end{equation}
This follows since a $0$-level $s$-border $(v,R)$ (with $v$ regular)
survives after removing $K$ if and only if $K\not\in \K_v$, and
a $1$-level $s$-border becomes a $0$-level $s$-border if and only 
if it is in conflict with $K$.  Combining this inequality with 
(\ref{Eq:ChargeEdgesHigh}), we get
\begin{eqnarray*}
\lefteqn{\frac{1}{n}N_0^{(s)}(\K) \le \frac{1}{n}N_1^{(s)}(\K)
+ O(n^{2d-4}\log^{s-1}n) \le} \\ & &
{\bf E}\left\{N_0^{(s)}(\R)\right\} - 
\frac{n-2d+2}{n}N_0^{(s)}(\K) + O(n^{2d-4}\log^{s-1}n) ,
\end{eqnarray*}
or
$$
\frac{n-2d+3}{n}N_0^{(s)}(\K) \le 
{\bf E}\left\{N_0^{(s)}(\R)\right\} + O(n^{2d-4}\log^{s-1}n) .
$$
Replacing $N_0^{(s)}(\K)$ and $N_0^{(s)}(\R)$ by their respective
maximum possible values $N_0^{(s)}(n)$ and $N_0^{(s)}(n-1)$,
we get the recurrence
$$
\frac{n-2d+3}{n}N_0^{(s)}(n) \le N_0^{(s)}(n-1) + O(n^{2d-4}\log^{s-1}n) ,
$$
whose solution is easily seen to be $N_0^{(s)}(n) = O(n^{2d-3}\log^s n)$.
This establishes the induction step and thus proves the asserted bound. 
In particular, we have so far
$$
g_d(n)=O(n^{2d-3}\log^{d-1} n).
$$

\medskip
\noindent{\bf Improved bounds for $s\geq 2$.} 
As promised, we next refine the analysis, and show that 
\begin{equation}\label{Eq:BordersHighImprove}
N_0^{(s)}(\K)=O(n^{2d-3}\log n),
\end{equation}
for any $1\leq s\leq d-1$, by establishing a sharper variant of (\ref{Eq:ChargeNaiveHigh}).

As in the three-dimensional case, the weakness of the preceding
analysis lies in the random sampling inequality (\ref{Eq:HighCS}), 
or, more precisely, in the term $N_1^{(s)}(\K)/n$ thereof.

Specifically, if a $1$-level $s$-border $(v,R)$ is in conflict with
$w>1$ sets of $\K$ then removing any one of these sets will make
$(v,R)$ a $0$-level $s$-border, so the probability of this to happen
is $w/n$, which is significantly larger than the bound $1/n$ used in 
(\ref{Eq:HighCS}). As above, we refer to $w$ as the {\em weight}
of $(v,R)$.  We can therefore modify the definition of 
$N_1^{(s)}(\K)$ so a border of weight $w$ is counted $w$ times. 
The preceeding discussion ensures that (\ref{Eq:HighCS}) still 
holds in the new setting.

We proceed to prove \ref{Eq:BordersHighImprove} by induction on $s$. The base case $s=1$ has already been
analyzed, and we have shown that $N_0^{(1)}(n)=O(n^{2d-3}\log n)$.
Fix $2\leq s\leq d-1$, and suppose that we have already proved that
$N_0^{(s')}(\K)=O(n^{2d-3}\log n)$, for all $1\leq s'<s$.

The following lemma generalizes Lemma \ref{Thm:Refine} to 
arbitrary dimension $d\geq 4$.
\begin{lemma}\label{Thm:RefineHighDim}
(i) The number of $1$-level $2$-borders, having weight 
$1$ and charged by two $0$-level neighboring $2$-borders,
is $O(N_1^{(1)}(\K)+n^{2d-3})$. \\
(ii) For $s\ge 3$, there are no $1$-level $s$-borders incident
to a regular vertex, having weight $1$, and charged by $s$ 
$0$-level neighboring $s$-borders.
\end{lemma}
\noindent{\bf Proof of Lemma \ref{Thm:RefineHighDim}.}
The proof of (i) is very similar to the proof of
Lemma~\ref{Thm:Refine}, and will be briefly presented later, after we
prove (ii).

So we assume that $s\ge 3$.
Let $(v,R)$ be a $1$-level $s$-border which has weight $1$ and is
charged by $s$ $0$-level neighboring $s$-borders, so that $v$ is
regular.  Let $C_{i_1j_1},C_{i_2j_2},\ldots,C_{i_{d-1}j_{d-1}}$ be the 
$(d-2)$-spheres incident to $v$. Without loss of generality, assume
that $R = \bigcap_{q=1}^{s} C^+_{i_qj_q}$. Let $f$ be the $s$-cell
associated with $(v,R)$, and let $e_1,\ldots,e_s$ be the $s$ edges of
$f$ incident to $v$, so that, for each $k=1,\ldots,s$, the edge $e_k$
lies on the circle $\bigcap_{1\le q\le d-1,q\ne k} C_{i_qj_q}$.
For $k=1,\ldots,s$, let $v_k$ denote the other endpoint of $e_k$, 
and let $C_{a_kb_k}$ denote the (unique) great sphere incident to 
$v_k$ and not containing $e_k$. Assume, without loss of generality,
that $v$ lies in $C^-_{a_kb_k}$, and put 
$R_k = C^+_{a_kb_k} \cap\bigcap_{1\le q\le s,q\ne k} C^+_{i_qj_q}$.
By construction, the $s$ $s$-borders $(v_k,R_k)$, for $k=1,\ldots,s$,
are precisely those that charge $(v,R)$, so they are all regular 
$0$-level $s$-borders. 

Note that $(v,R)$ is in conflict with each of the sets
$K_{a_1},K_{b_1},\ldots,K_{a_s},K_{b_s}$ for which the corresponding
index $a_k$ or $b_k$ is not one of $i_1,j_1,\ldots,i_{d-1},j_{d-1}$.
Indeed, removing such a set $K_{a_k}$, say, eliminates the sphere
$C_{a_kb_k}$ and thereby exposes $v$ to the extended $2^{d-1-s}$ 
permutation cells that surround $v_k$, so that they are all now
contained in $R$, so $(v,R)$ becomes a $0$-level $s$-border. 
However, since the weight of $(v,R)$ is $1$, only one of these 
sets, call it $K_b$, can be in conflict with $(v,R)$
(so $b\notin \{ i_1,j_1,\ldots,i_{d-1},j_{d-1} \}$).
This, and the fact that each of the $v_k$'s is regular,
is easily seen to imply the following property:
For each $k$, one of $a_k,b_k$, say $a_k$, belongs to $\{i_k,j_k\}$,
and the other index $b_k$ is $b$.

Fix a pair of distinct vertices $v_k$, $v_{\ell}$, and denote by
$\Pi_k$ (resp., $\Pi_\ell$) the collection of the $2^{d-1-s}$
permutations induced by the permutation cells that surround $v_k$ 
(resp., $v_\ell$) and are contained in $R_k$ (resp., $R_\ell$). 
Any pair of permutations in $\Pi_k$ differ from each other only
by swaps of some of the pairs $(i_q,j_q)$, for $q=s+1,\ldots,d-1$.
Hence the indices of each of these pairs appear consecutively 
in any of these permutations, and the locations of these pairs are 
fixed for all permutations. The set $K_b$ appears, somewhere in
between these pairs, in a fixed location in all permutations. 
A similar property holds for the permutations in $\Pi_\ell$.

Fix a permutation $\pi\in \Pi_k$. It has a ``twin'' permutation $\pi'$
in $\Pi_\ell$, in which the order of the two indices in each
of the pairs $(i_q,j_q)$, for $q=s+1,\ldots,d-1$, is the same as their
order in $\pi$. To gain more insight into the structure of $\pi$ and $\pi'$, let $\varphi$ and $\varphi'$ denote, respectively,
the permutation cells of $\A(\K)$ in which $\pi$ and $\pi'$ are 
generated. We can get from $\varphi$ to $\varphi'$ by first
crossing $C_{i_kb}$ into a corresponding $(d-1)$-cell $\varphi_0$
surrounding $f$ and then cross $C_{i_\ell b}$ into $\varphi'$.
This means that $\prec_{\varphi}$ and $\prec_{\varphi'}$ (i.e., $\pi$ and $\pi'$) are obtained
from each other by first swapping $K_b$ with $K_{i_k}$ and then by
swapping $K_b$ with $K_{i_\ell}$. As is easily checked, this implies
that $K_{i_k}$ and $K_{i_\ell}$ must be adjacent in $\pi$ and in
$\pi'$. This however cannot hold for {\em every} pair of distinct 
indices in $\{i_1,\ldots,i_s\}$ if $s\ge 3$. This contradiction
shows that for $s\ge 3$ there are no $1$-level $s$-borders which 
satisfy the assumptions in the lemma. This completes the proof 
of part (ii).

We now consider the case $s=2$, which, as noted above, can be handled 
in a manner that is very similar to the analysis in 
Lemma~\ref{Thm:Refine}. Specifically, let $(v,R)$ be a regular
$1$-level $2$-border of weight $1$ which is charged by two $0$-level 
$2$-borders $(v_1,R_1)$, $(v_2,R_2)$. (The number of degenerate
$1$-level $2$-borders is $O(n^{2d-3})$.) As in the proof of part (i), 
we may assume that both $v_1$ and $v_2$ are regular (for otherwise 
they would not charge $(v,R)$). Let 
$C_{i_1j_1},C_{i_2j_2},\ldots,C_{i_{d-1}j_{d-1}}\in \C(\K)$ be the 
$(d-2)$-spheres incident to $v$, and assume that 
$R= C_{i_1j_1}^+\cap C_{i_2j_2}^+$.  
Let $f$ be the 2-face associated with $(v,R)$.
For $k=1,2$, let $e_k$ denote the edge of $f$ incident to $v$ 
and contained in $C_{i_kj_k}$, and assume that $v_k$ is the 
other endpoint of $e_k$. Let $C_{a_kb_k}$ be the (unique) 
great sphere passing through $v_k$ and not containing $e_k$.

\begin{figure}[htbp]
\begin{center}
\input{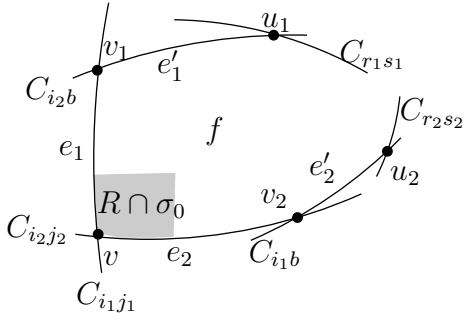}
\caption{\small\sf
The setup in the proof of Theorem \ref{Thm:RefineHighDim} 
for $s=2$: View within the sphere $\sigma_0=\bigcap_{q=3}^{d-1}C_{i_qj_q}$.} 
\label{Fig:RefineWeightHigh}
\end{center}
\end{figure}

As in the three-dimensional case, and similar to the preceding
analysis, since $(v,R)$ has weight $1$, the only case to be 
considered, up to symmetry, is where $a_1=i_2$, $a_2=i_1$, and 
$b_1=b_2=b$, where $b\neq \{i_1,j_1,\ldots,i_{d-1},j_{d-1}\}$.

The proof now continues as in the three-dimensional case, and 
we only provide a brief sketch of it. For $k=1,2$, we consider 
the other edge $e'_k$ of $f$ incident to $v_k$, denote by $u_k$ 
the other endpoint of $e_k$, and assume that neither of $u_1,u_2$ 
is degenerate. See Figure \ref{Fig:RefineWeightHigh}. 
We then consider the other great sphere $C_{r_ks_k}$ incident 
to $u_k$, for $k=1,2$, and distinguish between the following two cases:

\noindent
(a) $i_1\neq r_1,s_1$ or $i_2\neq r_2,s_2$. In the former case, 
removing $K_{i_1}$ leaves $u_1$ intact and extends $e'_1$ into 
a $0$-level $1$-border; the proof is argued exactly as in the 
three-dimensional case. The latter case is handled symmetrically, 
and we conclude that the number of $2$-borders of this kind is 
$O(N_1^{(1)}(\K))$.

\noindent
(b) $i_1=r_1$ and $i_2=r_2$ (or any of the symmetric pairs of
equalities). In this case we consider the $2$-sphere
$\sigma_0=\bigcap_{q=3}^{d-1}C_{i_qj_q}$ which contains $f$, and 
construct in it the arrangement $\A_{i_1}^{(\sigma_0)}$, formed 
by the circles $C_{i_1x}\cap \sigma_0$, for 
$x\notin \{i_1\}\cup\{i_3,j_3,\ldots,i_{d-1},j_{d-1}\}$. 
We note that $f$ is contained in a face $f_0$ of 
$\A_{i_1}^{(\sigma_0)}$ and touches its boundary at three 
distinct edges. This allows us to bound the number
of $2$-borders under consideration by $O(n^3)$, for a fixed choice of
$i_3,j_3,\ldots,i_{d-1},j_{d-1}$, arguing exactly as in the
three-dimensional case. In total, the number of these $2$-borders is
$O(n^3\cdot n^{2d-6})=O(n^{2d-3})$. This completes the proof of the lemma.
\hfill $\square$

\medskip
First, for $s=2$, we bound the quantity $N_1^{(1)}(\K)$ using the Clarkson-Shor analysis 
technique \cite{CS}, as we did in the proof of Lemma \ref{Thm:Refine}. 
That is, since each $1$-level $1$-border is defined by at most $2d-2$ sets 
of $\K$ and becomes a $0$-level $1$-border when we remove (at least) 
one set from $\K$, the number of $1$-level $1$-borders is 
$O\left({\bf E}\{N_0^{(1)}(\K')\}\right)$, where $\K'$ is a random 
sample of $n/2$ sets of $\K$.  Thus, combining this with the 
bound already established for $s=1$, we have
\begin{equation}\label{Eq:HighClarksonShor}
N_1^{(1)}(\K)=O\left({\bf E}\{N_0^{(1)}(\K')\}\right) =
O\left(N_0^{(1)}(n/2)\right)=O(n^{2d-3}\log n).
\end{equation}

With these preparations, we are now ready to complete the 
induction step for $s$.

Let $N_{1,1}^{(s)}(\K)$ denote the number of 1-level $s$-borders
having weight $1$, and let $N_{1,2}^{(s)}(\K)$ denote the number of
1-level $s$-borders having weight \textit{at least} $2$. Since a
1-level $s$-border of weight $w_i$ contributes to $N_{1}^{(s)}(\K)$
$w_i$ units, we have

\begin{equation}\label{Eq:HighDimWeightedSum}
N_1^{(s)}(\K)\geq N_{1,1}^{(s)}(\K)+2N_{1,2}^{(s)}(\K).
\end{equation}

Recall that we charge every $0$-level $s$-border (falling into subcase
(ii)) to $s$ neighboring 1-level $s$-borders. By Lemma
\ref{Thm:RefineHighDim} (and (\ref{Eq:HighClarksonShor})), all but
$O(n^{2d-3}\log n)$ 1-level $s$-borders, that have weight $1$, are
charged by at most $s-1$ neighboring 0-level $s$-borders. 
(This is the situation for $s=2$; the bound drops to $O(n^{2d-3})$ for $s\ge 3$.) Thus, we obtain the following refinement of \ref{Eq:ChargeNaiveHigh}:
\begin{equation}\label{Eq:HighDimChargingRefine}
sN_0^{(s)}(\K)\leq (s-1)N_{1,1}^{(s)}(\K)+sN_{1,2}^{(s)}(\K)+O(n^{2d-3}\log n).
\end{equation}

The combination of (\ref{Eq:HighDimWeightedSum}) and
(\ref{Eq:HighDimChargingRefine}), and the assumption that $s\geq 2$ (so $s/(s-1)\geq 2$) imply that,
$$
\frac{s}{s-1}N_0^{(s)}(\K)\leq N_1^{(s)}(\K)+O(n^{2d-3}\log n).
$$

Substituting $t=\frac{s}{s-1}-1=\frac{1}{s-1}>0$ and combining 
this with (\ref{Eq:HighCS}), we get
$$
\frac{1+t}{n}N_0^{(s)}(\K) \le \frac{1}{n}N_1^{(s)}(\K)
+ O(n^{2d-4}\log n)$$
$$
\le
{\bf E}\left\{N_0^{(s)}(\R)\right\} - 
\frac{n-2d+2}{n}N_0^{(s)}(\K) + O(n^{2d-4}\log n) ,
$$
or
$$
\frac{n-2d+3+t}{n}N_0^{(s)}(\K) \le 
{\bf E}\left\{N_0^{(s)}(\R)\right\} + O(n^{2d-4}\log n) .
$$
Thus, as above, we get the following recurrence
$$
\frac{n-2d+3+t}{n}N_0^{(s)}(n) \le N_0^{(s)}(n-1) + O(n^{2d-4}\log n) ,
$$
whose solution is easily seen to be 
$$
N_0^{(s)}(n) = O(n^{2d-3}\log n)
$$
(see, e.g., \cite[Proposition 3.1]{Tagansky}), which readily implies Theorem \ref{Thm:MainHigh}.
This completes the induction step and thus establishes \ref{Eq:BordersHighImprove} for all $s$.
We thus obtain the main result of the paper.
\begin{theorem}\label{Thm:MainHigh}
Any collection $\K$ of $n$ pairwise disjoint convex sets in 
$\reals^d$, for any $d\ge 3$,
admits at most $O(n^{2d-3}\log n)$ geometric permutations.
\end{theorem}

\section{Discussion}\label{Sec:Discuss}

Although the improvement presented in this paper is significant,
especially since no progress was made on the problem during the past 20 years, it is far
from satisfactory, since we strongly believe (and tend to conjecture)
that the correct upper bounds are close to $O(n^{d-1})$, for any 
$d\ge 3$. Improving further the bounds is the main open problem 
left by this study.
A modest subgoal is to get rid of the logarithmic factor in our bounds, and show, e.g., that $g_3(n)=O(n^3)$.

\end{document}